\UseRawInputEncoding
\documentclass[11pt, reqno]{amsart}
\usepackage{amsfonts,latexsym,enumerate}
\usepackage{amsmath}
\usepackage{amscd}
\usepackage{float,amsmath,amssymb,mathrsfs,bm,multirow,graphics}
\usepackage[dvips]{graphicx}
\usepackage[percent]{overpic}
\usepackage[numbers,sort&compress]{natbib}
\usepackage{todonotes}

\usepackage{tikz}

\newcommand{\R}{{\Bbb R}}

\newcommand{\C}{{\Bbb C}}

\newcommand{\Z}{{\Bbb Z}}

\newcommand{\res}{\text{\upshape Res\,}}

\newcommand{\re}{\text{\upshape Re\,}}
\newcommand{\im}{\text{\upshape Im\,}}

\newcommand{\mKdV}{\text{\upshape mKdV}}
\newcommand{\KdV}{\text{\upshape KdV}}

\def\XXint#1#2#3{{\setbox0=\hbox{$#1{#2#3}{\int}$}
\vcenter{\hbox{$#2#3$}}\kern-.5\wd0}}

\allowdisplaybreaks

\newtheorem{theorem}{Theorem}[section]

\newtheorem{lemma}[theorem]{Lemma}

\newtheorem{assumption}[theorem]{Assumption}
\newtheorem{remark}[theorem]{Remark}

\newtheorem{figuretext}[theorem]{Figure}
\newtheorem{RHproblem}[theorem]{RH problem}

\numberwithin{equation}{section}

\usepackage[colorlinks=true]{hyperref}
\hypersetup{urlcolor=blue, citecolor=red, linkcolor=blue}

\input epsf
\title[Miura transformation for the ``good'' Boussinesq equation]
{Miura transformation for the \\ ``good'' Boussinesq equation}

\author{C. Charlier}
\address{CC: Centre for Mathematical Sciences, Lund University, 221 00 Lund, Sweden.}
\email{christophe.charlier@math.lu.se}

\author{J. Lenells}
\address{JL: Department of Mathematics, KTH Royal Institute of Technology, 100 44 Stockholm, Sweden.}
\email{jlenells@kth.se}

\dedicatory{Dedicated to Professor A. S. Fokas on the occasion of his 70th birthday}

\begin{document}

\begin{abstract}
It is well-known that each solution of the mKdV equation gives rise, via the Miura transformation, to a solution of the KdV equation. In this work, we show that a similar Miura-type transformation exists also for the ``good'' Boussinesq equation. This transformation maps solutions of a second-order equation to solutions of the fourth-order Boussinesq equation. Just like in the case of mKdV and KdV, the correspondence exists also at the level of the underlying Riemann--Hilbert problems and this is in fact how we construct the new transformation.
\end{abstract}

\maketitle




\section{Introduction}
The Boussinesq equation was derived in \cite{B1872} as a model for small-amplitude waves propagating in shallow water. There are two different versions of the Boussinesq equation, often referred to as the ``good'' and ``bad'' Boussinesq equations \cite{M1981}; this work focuses on the ``good'' Boussinesq equation given by
\begin{align}\label{goodboussinesq}
u_{tt} - u_{xx} + (u^2)_{xx} + u_{xxxx} = 0.
\end{align}
Equation (\ref{goodboussinesq}) is relevant for multiple applications \cite{B1872, FST1983, S1974, ZK1965, Z1974} and possesses a rich mathematical structure, partly due to its integrability \cite{Z1974, ZS1974, H1973}.
The well-posedness of (\ref{goodboussinesq}) has been thoroughly studied (see e.g. \cite{KT2010, BS1988, CT2017, F2009, HM2015, L1993}) and asymptotic questions have also been investigated \cite{F2008, L1997, LS1995, CLW2022, WZ2022}.

In this paper, we construct a Miura-type transformation for (\ref{goodboussinesq}).
Our main result is the discovery that each complex-valued function $q$ satisfying
\begin{align}\label{3x3eq}
 iq_t - \frac{1}{\sqrt{3}}q_{xx} + 2\sqrt{3}\bar{q}\bar{q}_x =0,  
\end{align}
gives rise, via an explicit non-linear transformation, to a solution $u$ of \eqref{goodboussinesq}. More precisely, we have the following.

\begin{theorem}[Miura transformation for the ``good'' Boussinesq equation]\label{miurath}
If $q(x,t)$ is a complex-valued solution of \eqref{3x3eq}, then
\begin{align}\label{miuragoodintro}
u(x,t) = - \frac{4}{\sqrt{3}} \bigg( \frac{9}{2}\Big|q\Big(\frac{x}{3^{1/4}},t\Big)\Big|^2 + 3 \, \re\Big[e^{\frac{2\pi i}{3}}  q_x\Big(\frac{x}{3^{1/4}},t\Big)\Big] \bigg) + \frac{1}{2}
\end{align}
is a real-valued solution of the ``good'' Boussinesq equation \eqref{goodboussinesq}.
\end{theorem}
Two concrete examples of applications of the transformation \eqref{miuragoodintro} are provided in Section \ref{section:example}.

The assertion of Theorem \ref{miurath} can be verified by a long but straightforward computation. However, as we will explain in this paper, the Miura transformation (\ref{miuragoodintro}) can also be derived in a systematic manner. 
The underlying idea is that, thanks to its integrability, the ``good'' Boussinesq equation has a Riemann--Hilbert (RH) problem associated to it. This RH problem is singular and, after a transformation, this singular behavior manifests itself as a double pole in the solution at $k = 0$, where $k$  is the spectral parameter. Associated to the singular RH problem is a regular RH problem, which has the same jumps as the singular problem but no pole at $k = 0$. The integrable equation associated to this regular problem turns out to be equation (\ref{3x3eq}). As we describe in this paper, it is possible to relate the solutions of the regular and singular RH problems explicitly. This leads to an explicit relation between the associated integrable equations, which is exactly the transformation (\ref{miuragoodintro}).

\subsection{Background}
In 1968, Miura \cite{M1968} introduced an explicit transformation that maps solutions $q$ of the modified Korteveg--de Vries (mKdV) equation 
\begin{align}\label{mKdV}
 \mbox{mKdV: } \quad  q_{t} + q_{xxx} - 6q^{2}q_{x} = 0
\end{align}
to solutions $u$ of the Korteveg--de Vries (KdV) equation 
\begin{align}\label{KdV}
\mbox{KdV: } \quad & u_{t}+u_{xxx}-6uu_{x} = 0.
\end{align} 
More precisely, if $q$ satisfies \eqref{mKdV}, then 
\begin{align}\label{originalmiura}
u = q_{x} + q^{2}
\end{align}
satisfies (\ref{KdV}). The mapping (\ref{originalmiura}) from $q$ to $u$ is the Miura transformation. The Miura transformation was instrumental in the early days of the development of a theory for nonlinear integrable PDEs and has many applications; for example, it was used by Miura, Gardner, and Kruskal to construct infinitely many conservation laws for the KdV equation \cite{MGK1968}. 
In fact, the transformation (\ref{originalmiura}) is a reflection of a much deeper correspondence between the mKdV and KdV equations that extends to the associated RH formulations. The present paper was inspired by work of Fokas and Its \cite{FI1994} where this correspondence is discussed, see \cite[Proposition 2.3 and Remark 2.5]{FI1994}. Roughly speaking, the correspondence exists because the RH problem associated to KdV is singular and the associated regular RH problem (i.e., the RH problem with the same jumps, but with no singularity) is the RH problem for mKdV, see Section \ref{kdvsec} for details.
It is possible to relate the solutions of the regular and singular RH problems explicitly and this is one way to arrive at the Miura transformation (\ref{originalmiura}).

The relationship between the RH problems for mKdV and KdV is analogous to the relationship between the RH problems for equations (\ref{goodboussinesq}) and (\ref{3x3eq}); the only difference is that whereas the solution of the RH problem for Boussinesq is a $3 \times 3$-matrix-valued function with a double pole at the origin, the solution of the RH problem for KdV is a $2 \times 2$-matrix-valued function with a simple pole at the origin. Other than that the constructions underlying the Miura transformation (\ref{originalmiura}) and the transformation (\ref{miuragoodintro}) are conceptually identical; this is why we refer to (\ref{miuragoodintro}) as a Miura transformation. 

Equation (\ref{3x3eq}) was listed as an integrable system obtained by reduction in \cite{Mik1981} and was, as far as we know, first studied in \cite{L3x3}. In \cite{CL2021}, a Riemann--Hilbert representation for the solution of the initial value problem for (\ref{3x3eq}) was derived and used to obtain formulas for the long-time asymptotics. 
Other works analyzing aspects of $3 \times 3$ RH problems for integrable evolution equations include \cite{BSZ2016, CIL2010, DTT1982, BLS2017, HL2020, LGX2018, CLW2022, BS2013, XF2016}.

\subsection{Outline of the paper}
Our main result is stated in Section \ref{mainsec}. It describes the correspondence between the RH problems associated to (\ref{goodboussinesq}) and (\ref{3x3eq}), and how it leads to the transformation (\ref{miuragoodintro}). Several basic properties of the solution of the RH problems for (\ref{goodboussinesq}) and (\ref{3x3eq}) are derived in Sections \ref{Msec} and \ref{msec}, respectively. The proof of the main theorem is presented in Section \ref{proofsec} and two illustrating examples are given in Section \ref{section:example}. A detailed account of the correspondence between the RH problems associated to mKdV and KdV, and how it leads to the Miura transformation (\ref{originalmiura}), is included in Section \ref{kdvsec}.
Appendix \ref{initialapp} reviews some facts related to the initial value problems for the Boussinesq  equation and equation (\ref{3x3eq}) on the real line. In Appendix \ref{solitonapp}, we consider soliton solutions of (\ref{3x3eq}).

\section{Main result}\label{mainsec}
Before we can state our main result, we need to introduce RH problems associated to (\ref{goodboussinesq}) and (\ref{3x3eq}).
Instead of considering a RH problem directly related to (\ref{goodboussinesq}), it is convenient to first transform (\ref{goodboussinesq}) into the equation
\begin{align}\label{goodboussinesqnouxx}
& u_{tt} + \frac{1}{3} u_{xxxx} + \frac{4}{3} (u^2)_{xx}  = 0
\end{align}
by means of the transformation
\begin{align}\label{fromboussinesqtogoodboussinesq}
  \hat{u}(x,t) = \frac{4}{\sqrt{3}}u\bigg(\frac{x}{3^{1/4}},t\bigg) + \frac{1}{2}.
\end{align}
It is easy to see that if $u$ satisfies (\ref{goodboussinesqnouxx}), then $\hat{u}$ satisfies (\ref{goodboussinesq}).
Following \cite{Z1974, DTT1982}, we rewrite (\ref{goodboussinesqnouxx}) as the system
\begin{align}\label{boussinesqsystem}
& \begin{cases}
 v_{t} + \frac{1}{3}u_{xxx} + \frac{4}{3}(u^{2})_{x} = 0,
 \\
 u_t = v_x.
\end{cases}
\end{align}
A Riemann--Hilbert approach to the system (\ref{boussinesqsystem}) was developed in \cite{CL2022} and we will work with the RH problem for (\ref{boussinesqsystem}) constructed in \cite{CL2022}. This RH problem has a jump contour $\Gamma$ consisting of six rays oriented away from the origin as in Figure \ref{Gamma.pdf}.
Let $\omega := e^{\frac{2\pi i}{3}}$ and define $\{l_j(k), z_j(k)\}_{j=1}^3$ by
\begin{align}\label{lmexpressions}
&l_j(k) = \omega^j k, \quad z_j(k) = \omega^{2j} k^{2}, \qquad k \in \C.
\end{align}
For $1 \leq i \neq j \leq 3$, let $\theta_{ij} = \theta_{ij}(x,t,k)$ be given by
$$\theta_{ij}(x,t,k) = (l_i - l_j)x + (z_i - z_j)t.$$
Given two functions $r_1:(0,\infty) \to \C$ and $r_2:(-\infty,0) \to \C$, we define the jump matrix $v(x,t,k)$ for $k \in \Gamma$ by
\begin{align}\nonumber
&  v_1 = 
  \begin{pmatrix}  
 1 & - r_1(k)e^{-\theta_{21}} & 0 \\
  r_1^*(k)e^{\theta_{21}} & 1 - |r_1(k)|^2 & 0 \\
  0 & 0 & 1
  \end{pmatrix},
&&  v_2 = 
  \begin{pmatrix}   
 1 & 0 & 0 \\
 0 & 1 - |r_2(\omega k)|^2 & -r_2^*(\omega k)e^{-\theta_{32}} \\
 0 & r_2(\omega k)e^{\theta_{32}} & 1 
    \end{pmatrix},
   	\\ \nonumber
  &v_3 = 
  \begin{pmatrix} 
 1 - |r_1(\omega^2 k)|^2 & 0 & r_1^*(\omega^2 k)e^{-\theta_{31}} \\
 0 & 1 & 0 \\
 -r_1(\omega^2 k)e^{\theta_{31}} & 0 & 1  
  \end{pmatrix}, &&  v_4 = 
  \begin{pmatrix}  
  1 - |r_2(k)|^2 & -r_2^*(k) e^{-\theta_{21}} & 0 \\
  r_2(k)e^{\theta_{21}} & 1 & 0 \\
  0 & 0 & 1
   \end{pmatrix},
   	\\ \label{vdef}
&  v_5 = 
  \begin{pmatrix}
  1 & 0 & 0 \\
  0 & 1 & -r_1(\omega k)e^{-\theta_{32}} \\
  0 & r_1^*(\omega k)e^{\theta_{32}} & 1 - |r_1(\omega k)|^2
  \end{pmatrix},
&& v_6 = 
  \begin{pmatrix} 
  1 & 0 & r_2(\omega^2 k)e^{-\theta_{31}} \\
  0 & 1 & 0 \\
  -r_2^*(\omega^2 k)e^{\theta_{31}} & 0 & 1 - |r_2(\omega^2 k)|^2
   \end{pmatrix},
\end{align}
where $v_j$ denotes the restriction of $v$  to the subcontour of $\Gamma$ labeled by $j$ in Figure \ref{Gamma.pdf}.
Let $D_n$, $n = 1, \dots, 6$, be the open subsets displayed in Figure \ref{Gamma.pdf}.

We can now state the RH problem for (\ref{boussinesqsystem}). We denote the solutions of the RH problems associated to (\ref{3x3eq}) and (\ref{boussinesqsystem}) by $m$ and $M$, respectively.

\begin{figure}
\begin{center}
 \begin{overpic}[width=.5\textwidth]{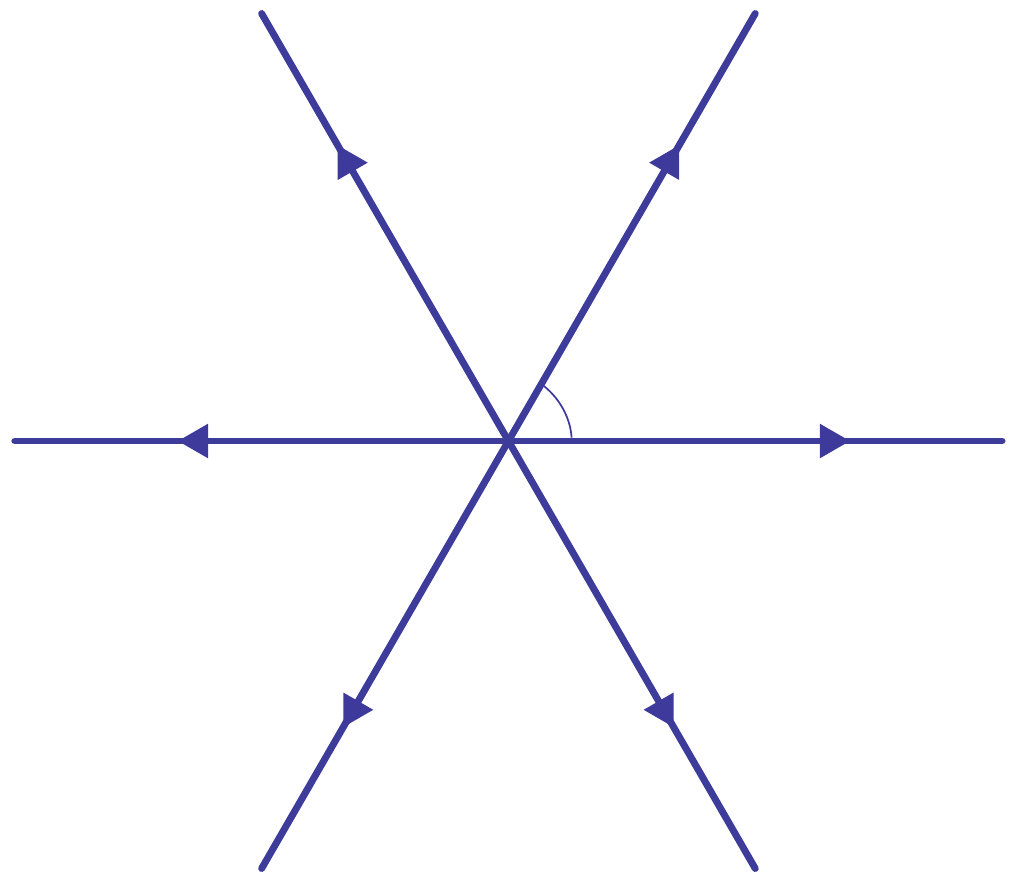}
  \put(101,42.5){\small $\Gamma$}
 \put(57,47){\small $\pi/3$}
 \put(80,60){\small $D_1$}
 \put(48,74){\small $D_2$}
 \put(17,60){\small $D_3$}
 \put(17,25){\small $D_4$}
 \put(48,12){\small $D_5$}
 \put(80,25){\small $D_6$}
  \put(81,38){\small $1$}
 \put(68,69){\small $2$}
 \put(30,69){\small $3$}
 \put(18,38){\small $4$}
 \put(30,16){\small $5$}
 \put(67,16){\small $6$}
   \end{overpic}
     \begin{figuretext}\label{Gamma.pdf}
       The contour $\Gamma$ and the open sets $D_n$, $n = 1, \dots, 6$, which decompose the complex $k$-plane.
     \end{figuretext}
     \end{center}
\end{figure}

\begin{RHproblem}[RH problem for $M$]\label{RHgoodboussinesq}
Let $r_1:(0,\infty) \to \C$ and $r_2:(-\infty,0) \to \C$ be given functions. Find a $3 \times 3$-matrix valued function $M(x,t,k)$ with the following properties:
\begin{enumerate}[(a)]
\item $M(x,t,\cdot) : \mathbb{C}\setminus \Gamma \to \mathbb{C}^{3 \times 3}$ is analytic.

\item The limits of $M(x,t,k)$ as $k$ approaches $\Gamma\setminus \{0\}$ from the left $(+)$ and right $(-)$ exist, are continuous on $\Gamma\setminus \{0\}$, and satisfy the jump condition
\begin{align}\label{Mjumpcondition}
& M_{+}(x,t,k) = M_{-}(x,t,k)v(x,t,k), \qquad k \in \Gamma \setminus \{0\},
\end{align}
where $v$ is defined in terms of $r_{1}$ and $r_{2}$ by \eqref{vdef}.

\item As $k \to \infty$, $k \notin \Gamma$, we have
\begin{align}\label{Matinfty}
M(x,t,k) = I + \frac{M^{(1)}(x,t)}{k} + \frac{M^{(2)}(x,t)}{k^{2}} + O\bigg(\frac{1}{k^3}\bigg),
\end{align}
where the matrices $M^{(1)}$ and $M^{(2)}$ depend on $x$ and $t$ but not on $k$, and satisfy
\begin{align}\label{singRHMatinftyb}
M_{12}^{(1)} = M_{13}^{(1)} = M_{12}^{(2)} + M_{13}^{(2)} = 0.
\end{align}

\item There exist matrices $\{\mathcal{M}_1^{(l)}(x,t)\}_{l=-2}^{+\infty}$ such that, for any $N \geq -2$,
\begin{align}\label{singRHMat0}
M(x,t,k) = \sum_{l=-2}^{N} \mathcal{M}_1^{(l)}(x,t)k^{l} + O(k^{N+1}) \qquad \text{as}\ k \to 0, \ k \in \bar{D}_1.
\end{align}
Moreover, there exist complex coefficients $\alpha, \beta, \gamma, \delta, \epsilon$ depending on $x$ and $t$ such that
\begin{align} \label{explicitMcalpm2p}
\mathcal{M}_{1}^{(-2)}(x,t) & = \alpha(x,t) \begin{pmatrix}
\omega & 0 & 0 \\
\omega & 0 & 0 \\
\omega & 0 & 0
\end{pmatrix}, 
	\\
\mathcal{M}_{1}^{(-1)}(x,t) & = \beta(x,t) \begin{pmatrix}
\omega^{2} & 0 & 0 \\
\omega^{2} & 0 & 0 \\
\omega^{2} & 0 & 0 
\end{pmatrix} + \gamma(x,t) \begin{pmatrix}
\omega^{2} & 0 & 0 \\
1 & 0 & 0 \\
\omega & 0 & 0
\end{pmatrix}  + \delta(x,t) \begin{pmatrix}
0 & 1-\omega & 0 \\
0 & 1-\omega & 0 \\
0 & 1-\omega & 0
\end{pmatrix}, 
\end{align}
and the third column of $\mathcal{M}_{1}^{(0)}(x,t)$, denoted by $[\mathcal{M}_{1}^{(0)}(x,t)]_{3}$, is given by
\begin{align}
[\mathcal{M}_{1}^{(0)}(x,t)]_{3} = \epsilon(x,t) \begin{pmatrix}
1 \\
1 \\
1 
\end{pmatrix}. \label{explicit Mcalp0p third column}
\end{align}

\item $M$ satisfies the symmetries
\begin{align}\label{singRHsymm}
M(x,t, k) = \mathcal{A} M(x,t,\omega k)\mathcal{A}^{-1} = \mathcal{B} \overline{M(x,t,\overline{k})}\mathcal{B}, \qquad k \in \C \setminus \Gamma,
\end{align}
where
\begin{align}\label{AcalBcaldef}
\mathcal{A} := \begin{pmatrix}
0 & 0 & 1 \\
1 & 0 & 0 \\
0 & 1 & 0
\end{pmatrix} \qquad \mbox{ and } \qquad \mathcal{B} := \begin{pmatrix}
0 & 1 & 0 \\
1 & 0 & 0 \\
0 & 0 & 1
\end{pmatrix}.
\end{align}
\end{enumerate}
\end{RHproblem}

\begin{remark}
We show in Section \ref{Msec} that the solution $M$ of RH problem \ref{RHgoodboussinesq} is unique if it exists. 
Moreover, we review in Appendix \ref{initialapp} how the initial value problem for (\ref{boussinesqsystem}) can be solved in terms of the solution $M$ of RH problem \ref{RHgoodboussinesq}; in the context of the initial value problem, the functions $r_1(k)$ and $r_2(k)$ are defined in terms of the given initial data $q_0(x) = q(x,0)$ and can be thought of as nonlinear Fourier transforms of $q_0$. For the present purposes, there is no need to assume that $r_1(k)$ and $r_2(k)$ originate from any initial data. 
\end{remark}

Note that $M(x,t,k)$ has a double pole at $k = 0$ as a consequence of (\ref{singRHMat0}) in the generic case when $\alpha$ is non-zero.
The RH problem associated to \eqref{3x3eq} has the same form as RH problem \ref{RHgoodboussinesq} except that its solution $m(x,t,k)$ does not have a singularity at the origin.

\begin{RHproblem}[RH problem for $m$]\label{RH3x3}
Let $r_1:(0,\infty) \to \C$ and $r_2:(-\infty,0) \to \C$ be some given functions. Find a $3 \times 3$-matrix valued function $m(x,t,k)$ with the following properties:
\begin{enumerate}[(a)]
\item $m(x,t,\cdot) : \mathbb{C}\setminus \Gamma \to \mathbb{C}^{3 \times 3}$ is analytic.

\item The limits of $m(x,t,k)$ as $k$ approaches $\Gamma\setminus \{0\}$ from the left $(+)$ and right $(-)$ exist, are continuous on $\Gamma\setminus \{0\}$, and satisfy the jump condition
\begin{align}\label{mjumpcondition}
& m_{+}(x,t,k) = m_{-}(x,t,k)v(x,t,k), \qquad k \in \Gamma \setminus \{0\},
\end{align}
where $v$ is defined in terms of $r_{1}$ and $r_{2}$ by \eqref{vdef}.
\item As $k \to \infty$, $k \notin \Gamma$, we have
\begin{align}\label{asymp m at inf in RH def}
m(x,t,k) = I + O(k^{-1}).
\end{align}
\item There exist matrices $\{\mathfrak{m}_1^{(l)}(x,t)\}_{l=0}^{+\infty}$ such that, for any $N \geq 0$,
\begin{align}
m(x,t,k) = \sum_{l=0}^{N} \mathfrak{m}_1^{(l)}(x,t)k^{l} + O(k^{N+1}) \qquad \text{as}\ k \to 0, \ k \in \bar{D}_1.
\end{align}
\end{enumerate}
\end{RHproblem}

\begin{remark}
We show in Section \ref{msec} that the solution $m$ of RH problem \ref{RH3x3} is unique if it exists, and that $m$ automatically satisfies the same symmetries (\ref{singRHsymm}) as the solution $M$ of RH problem \ref{RHgoodboussinesq}; we have included (\ref{singRHsymm}) as part of the formulation of RH problem \ref{RHgoodboussinesq} to ensure uniqueness of the solution $M$.
\end{remark}

Our main result is formulated under the the following two assumptions on the spectral functions $r_1(k)$ and $r_2(k)$.  
 
\begin{assumption}\label{r1r2assumption}
Assume $r_{1}$ and $r_{2}$ satisfy the following properties:
\begin{enumerate}[$(i)$]
\item $r_1 \in C^\infty((0,\infty))$ and $r_2 \in C^\infty((-\infty,0))$. 
\item The functions $r_1(k)$, $r_2(k)$, and their derivatives $\partial_{k}^{j}r_{l}(k)$ have continuous boundary values at $k=0$ for $l = 1,2$ and for all $j = 0,1,2,\ldots$, and there exist expansions
 \begin{subequations}\label{r1r2atzero assumptions}
\begin{align}
& r_{1}(k) = r_{1}(0) + r_{1}'(0)k + \tfrac{1}{2}r_{1}''(0)k^{2} + \cdots, & & k \to 0, \ k >0, \\
& r_{2}(k) = r_{2}(0) + r_{2}'(0)k + \tfrac{1}{2}r_{2}''(0)k^{2} + \cdots, & & k \to 0, \ k <0,
\end{align}
\end{subequations}
which can be differentiated termwise any number of times.
\item $r_1(k)$ and $r_2(k)$ are rapidly decreasing as $k \to \pm \infty$, i.e., for each integer $N \geq 0$,
\begin{subequations}\label{r1r2rapiddecay assumptions}
\begin{align}
& \max_{j=0,1,\dots,N}\sup_{k \in (0,\infty)} (1+|k|)^N |\partial_k^jr_1(k)| < \infty,  
	\\
& \max_{j=0,1,\dots,N} \sup_{k \in (-\infty, 0)} (1+|k|)^N|\partial_k^jr_2(k)| < \infty.
\end{align}
\end{subequations} 
\end{enumerate}
\end{assumption}

\begin{assumption}\label{r1r2at0assumption}
Assume $r_{1}$ and $r_{2}$ satisfy
\begin{align*}
r_1(0) = \omega, \qquad r_2(0) = 1.
\end{align*}
\end{assumption} 

\begin{remark}
The RH problems \ref{RHgoodboussinesq} and \ref{RH3x3} can be formulated for any choice of the functions $\{r_j(k)\}_1^2$ and they give rise to solutions of (\ref{goodboussinesq}) and (\ref{3x3eq}), respectively, for a very large class of such functions. 
Here, we will restrict ourselves to functions $\{r_j(k)\}_1^2$ satisfying Assumptions \ref{r1r2assumption} and \ref{r1r2at0assumption}. These assumptions are natural from the point of view of the initial value problem for the ``good'' Boussinesq equation (\ref{boussinesqsystem}), in the sense that they are fulfilled for generic solitonless Schwartz class initial data, see Appendix \ref{initialapp}. 
\end{remark}

The following is our main result.

\begin{theorem}[Miura correspondence for ``good'' Boussinesq]\label{mainth}
Let $r_{1}(k)$ and $r_{2}(k)$ be such that Assumptions \ref{r1r2assumption} and  \ref{r1r2at0assumption} hold. 
Define the $3 \times 3$-matrix valued jump matrix $v(x,t,k)$ in terms of $r_{1}$ and $r_{2}$ by (\ref{vdef}). Let $\mathcal{N}$ be an open subset of $\mathbb{R}\times [0,+\infty)$ and suppose RH problem \ref{RH3x3} has a (necessarily unique) solution $m(x,t,\cdot)$ for each $(x,t) \in \mathcal{N}$.

\begin{enumerate}[$(a)$]
\item \label{partagood}
The function $q$ defined by
\begin{align}\label{recoverq3x3}
q(x,t) = \lim_{k\to \infty} \big(k \, m(x,t,k)\big)_{13}
\end{align}
is smooth and satisfies (\ref{3x3eq}) for $(x,t) \in \mathcal{N}$.

\item \label{partbgood}
The functions $m_{33}^{(1)}(x,t)$ and $m_{13}^{(2)}(x,t)$ are well-defined by 
\begin{subequations}\label{m133m213limm}
\begin{align}
& m_{33}^{(1)}(x,t) =  \lim_{k\to \infty} k\big(m_{33}(x,t,k) - 1\big), 
	\\
& m_{13}^{(2)}(x,t) =  \lim_{k\to \infty} k^2\bigg(m_{13}(x,t,k) - \frac{q(x,t)}{k}\bigg).
\end{align}
\end{subequations}

\item \label{partcgood}
Define $A(x,t)$ and $B(x,t)$ by
\begin{subequations}\label{BAdef}
\begin{align}
& B(x,t) = y_1(x,t) \begin{pmatrix} 
\omega & \omega^2 & 1 \\ 
\omega & \omega^2 & 1  \\ 
\omega & \omega^2 & 1 
\end{pmatrix},
	\\
& A(x,t) = y_2(x,t)\begin{pmatrix} 
\omega^2 & \omega & 1 \\ 
\omega^2 & \omega & 1 \\ 
\omega^2 & \omega & 1
\end{pmatrix}
+ y_3(x,t)\begin{pmatrix} \omega^2 & 1 & \omega \\ 
1 & \omega & \omega^2 \\ 
\omega & \omega^2 & 1
\end{pmatrix},
\end{align}
\end{subequations}
where $\{y_j(x,t)\}_1^3$ are real-valued and given by
\begin{align}\nonumber
y_1 = &\; \omega\overline{m_{13}^{(2)}}  + \omega^{2} m_{13}^{(2)}-(\omega\bar{q}+\omega^{2} q)m_{33}^{(1)}  + 2 |q|^2 
	\\\nonumber
& - (1 -\omega)\bar{q}^2 - (1-\omega^{2})q^2,
	\\\label{yjdef}
y_2 = &\; \frac{\omega q -  \bar{q}}{1-\omega}, \qquad
y_3 = \frac{\omega(q-\bar{q})}{1-\omega^2}.
\end{align}	
Then the $3\times 3$-matrix valued function $M(x,t,k)$ defined by
\begin{align}\label{Mfrommdefgood}
M(x,t,k) = \bigg(I + \frac{A(x,t)}{k} + \frac{B(x,t)}{k^2}\bigg)m(x,t,k)
\end{align}
satisfies RH problem \ref{RHgoodboussinesq} for each $(x,t) \in \mathcal{N}$.

\item \label{partdgood}
The function $u$ defined by
\begin{align}\label{recoverugood}
 u(x,t) = -\frac{3}{2}\frac{\partial}{\partial x}\lim_{k\to \infty}k\big(M_{33}(x,t,k) - 1\big),
\end{align}
is smooth and real-valued on $\mathcal{N}$ and satisfies the ``good'' Boussinesq equation (\ref{goodboussinesqnouxx}) for $(x,t) \in \mathcal{N}$. Moreover, if 
\begin{align}\label{vintermsofq}
v(x,t) = -\sqrt{3}  \, \im\big(6q^3 - 3q\bar{q}_x + 6\omega^2 q q_x + \omega q_{xx}\big),
\end{align}
then $u$ and $v$ satisfy the system (\ref{boussinesqsystem}) on $\mathcal{N}$.

\item \label{partegood}
The solutions $u$ and $q$ are related by the Miura-type transformation
\begin{align}\label{miuragood}
u(x,t) = -\frac{9}{2}|q(x,t)|^2 - 3 \, \re\big(\omega q_x(x,t)\big), \qquad (x,t) \in \mathcal{N}.
\end{align}

\end{enumerate}

\end{theorem}
\begin{proof}
See Section \ref{proofsec}.
\end{proof}

\begin{remark}
The Miura transformation of Theorem \ref{miurath} follows immediately from \eqref{miuragood} and \eqref{fromboussinesqtogoodboussinesq}.
\end{remark}

\section{Properties of the solution $M$ of RH problem \ref{RHgoodboussinesq}}\label{Msec}
We establish several properties of the solution $M$ of RH problem \ref{RHgoodboussinesq}.
Throughout this section, we assume that $\mathcal{N}$ is a (not necessarily open) subset of $\mathbb{R}\times [0,+\infty)$ and that $M(x,t,\cdot)$ satisfies RH problem \ref{RHgoodboussinesq} for each $(x,t) \in \mathcal{N}$.

\begin{lemma}[Asymptotics of $M$ as $k \to \infty$]\label{Matinftylemma}
For each $(x,t) \in \mathcal{N}$, $M(x,t,k)$ admits an expansion of the following form as $k \to \infty$:
\begin{align}
M(x,t,k) = & \; I + \frac{M_{33}^{(1)}}{k} \begin{pmatrix} \omega^2 & 0 & 0 \\ 
0 & \omega & 0 \\ 
0 & 0 & 1
\end{pmatrix} 
+ \frac{M_{33}^{(2)}}{k^2}\begin{pmatrix} 
\omega  & 0 & 0 \\ 
0 & \omega^2 & 0 \\ 
0 & 0 &  1
\end{pmatrix} \nonumber
	 \\ \label{singMatinfty} 
& \; + \frac{\widetilde{M}_{13}^{(2)}}{(1-\omega)k^2}\begin{pmatrix} 
0 & 1 & -1 \\ 
-\omega & 0 & \omega \\ 
\omega^2 & -\omega^2  &  0
\end{pmatrix} + O(k^{-3}), 
\end{align}
where $M_{33}^{(1)}(x,t)$, $M_{33}^{(2)}(x,t)$, and $\widetilde{M}_{13}^{(2)}(x,t) := -(1-\omega) M_{13}^{(2)}(x,t)$ are real-valued functions of $(x,t) \in \mathcal{N}$.
\end{lemma}
\begin{proof}
Combining the large $k$ asymptotics \eqref{Matinfty} with the symmetries (\ref{singRHsymm}), we obtain
\begin{align*}
& M^{(1)} = \omega^2 \mathcal{A}M^{(1)}\mathcal{A}^{-1} = \mathcal{B}\overline{M^{(1)}}\mathcal{B}, 
	\\
& M^{(2)} = \omega \mathcal{A}M^{(2)}\mathcal{A}^{-1} = \mathcal{B}\overline{M^{(2)}}\mathcal{B}.
\end{align*}
Recalling the definitions (\ref{AcalBcaldef}) of $\mathcal{A}$ and $\mathcal{B}$, and using the conditions in (\ref{singRHMatinftyb}), we find that
\begin{align}\label{singMatinfty proof}
M(x,t,k) = I + \frac{M_{33}^{(1)}}{k} \begin{pmatrix} \omega^2 & 0 & 0 \\ 
0 & \omega & 0 \\ 
0 & 0 & 1
\end{pmatrix} 
+ \frac{1}{k^2}\begin{pmatrix} \omega M_{33}^{(2)} & \omega^2 \overline{M_{13}^{(2)}} & M_{13}^{(2)} \\ 
\omega M_{13}^{(2)} & \omega^2 M_{33}^{(2)} & \overline{M_{13}^{(2)}} \\ 
\omega \overline{M_{13}^{(2)}} & \omega^2 M_{13}^{(2)} &  M_{33}^{(2)}
\end{pmatrix} + O(k^{-3}),
\end{align}
and that $M_{33}^{(1)},M_{33}^{(2)}$ are real-valued. The last condition in (\ref{singRHMatinftyb}) implies that $\omega^2 \overline{M_{13}^{(2)}} + M_{13}^{(2)} = 0$. Hence $(1-\omega) M_{13}^{(2)}$ is real-valued and the desired conclusion follows.
\end{proof}

\begin{lemma}[Reality of $\alpha, \beta, \gamma, \delta, \epsilon$]\label{realitylemma}
If $r_{1}$ and $r_{2}$ satisfy Assumptions \ref{r1r2assumption} and \ref{r1r2at0assumption}, then the functions $\alpha, \beta, \gamma, \delta, \epsilon$ in \eqref{explicitMcalpm2p}--\eqref{explicit Mcalp0p third column} are real-valued and
\begin{align}\label{r1primeat0expression}
& r_1'(0) = (1-\omega^{2}) \frac{x\alpha(x,t) + \beta(x,t) + \delta(x,t)}{\alpha(x,t)} \quad \text{for all $(x,t) \in \mathcal{N}$}.
\end{align}
\end{lemma}
\begin{proof}
We fix $(x,t) \in \mathcal{N}$ and omit the $(x,t)$-dependence for convenience. The main idea is to combine the jump relation (\ref{Mjumpcondition}) with the symmetries
\begin{align}\label{lol3}
M_{1}(k) = \mathcal{A}M_{3}(\omega k) \mathcal{A}^{-1} = \mathcal{B}\overline{M_6(\overline{k})}\mathcal{B}, \qquad k \in \bar{D}_{1},
\end{align}
where $M_n$, $n = 1, \dots, 6$, denotes the restriction of $M$ to $D_n$. The jumps for $M$ are not analytic in $k$, but they can be Taylor expanded as $k \to 0$. The existence of the all-order expansion \eqref{singRHMat0} of $M_{1}$ implies together with (\ref{Mjumpcondition}) the existence of similar expansions for $M_{2},\ldots,M_{6}$. We define the following formal series:
\begin{align}\label{def of vjL}
& M_{j}^{L}(k) := \sum_{l=-2}^{+\infty}\mathcal{M}_{j}^{(l)}(0)k^{l}, \qquad v_{j}^{L}(k) := \sum_{l=0}^{+\infty}v_{j}^{(l)}(0)k^{l}, & & j = 1,\ldots,6.
\end{align}
The matrix $M_{j}$ coincides to any order with $M_{j}^{L}$ as $k \to 0$, $k \in \bar{D}_{j}$, and the matrix $v_{j}$ coincides to any order with $v_{j}^{L}$ as $k \to 0$, $k \in \Gamma_{j}$. Using the symmetries \eqref{lol3}, we obtain
\begin{align}
M_{1}^{L}(k) & = \mathcal{A}M_{3}^{L}(\omega k) \mathcal{A}^{-1} = \mathcal{A}M_{1}^{L}(\omega k) v_{2}^{L}(\omega k) v_{3}^{L}(\omega k) \mathcal{A}^{-1} \nonumber \\
& = \mathcal{A}M_{1}^{L}(\omega k) \mathcal{A}^{-1} v_{6}^{L}(k)v_{1}^{L}(k), \label{M1 Acal rel} \\
M_1^{L}(k) & = \mathcal{B}\overline{M_1^{L}(\overline{k})}\overline{v_1^{L}(\overline{k})^{-1}}\mathcal{B}
= \mathcal{B}\overline{M_1^{L}(\overline{k})}\mathcal{B}v_1^{L}(k). \label{M1 Bcal rel}
\end{align}
Using \eqref{M1 Acal rel} and considering the terms of $O(k^{-2})$ and $O(k^{-1})$, we obtain
\begin{align}
& \mathcal{M}_{1}^{(-2)} = \omega \mathcal{A} \mathcal{M}_{1}^{(-2)}\mathcal{A}^{-1}(v_{6}v_{1})(0), \label{M1pm2p Arel}\\
& \mathcal{M}_{1}^{(-1)} = \omega^{2} \mathcal{A}\mathcal{M}_{1}^{(-1)}\mathcal{A}^{-1}(v_{6}v_{1})(0) + \omega \mathcal{A}\mathcal{M}_{1}^{(-2)}\mathcal{A}^{-1}(v_{6}v_{1})'(0).\label{M1pm1p Arel} 
\end{align}
On the other hand, Taylor expanding the jump condition $M_+ = M_- v$ along each of the six rays of $\Gamma$, we find the relations
\begin{align}\label{product of the jumps is identity}
& (v_1v_2v_3v_4v_5v_6)(x,t,0) =I, \qquad \partial_{k}^j(v_1v_2v_3v_4v_5v_6)(x,t,0) = 0, \qquad j = 1,2, \dots
\end{align}
These identities give rise to relations between the values of $r_1$, $r_2$, and their derivatives at $k = 0$. In particular, since $r_{1}$ and $r_{2}$ satisfy Assumption \ref{r1r2at0assumption},  \eqref{product of the jumps is identity} is satisfied with $j=1$ and $j=2$ if and only if
\begin{align}
& \overline{r_{1}'(0)} = - \omega r_{1}'(0), \label{identity for r1'} \\
& \overline{r_{2}'(0)} = -r_{2}'(0), \quad \mbox{i.e. }r_2'(0) \in i\R, \label{identity for r2'} \\
& \omega \overline{r_{1}''(0)}+ \overline{r_{2}''(0)} = -\big(r_{1}''(0) + \omega r_{2}''(0)\big) + 2 \omega r_{1}'(0)r_{2}'(0). \label{identity for r1''r2''}
\end{align}
Substituting \eqref{identity for r1'} into the first column of \eqref{M1pm1p Arel}, we deduce \eqref{r1primeat0expression}. 
On the other hand, the first three orders of  \eqref{M1 Bcal rel} yield
\begin{align}
& \mathcal{M}_{1}^{(-2)} = \mathcal{B} \overline{\mathcal{M}_{1}^{(-2)}}\mathcal{B}v_{1}(0), \label{M1pm2p Brel}\\
& \mathcal{M}_{1}^{(-1)} = \mathcal{B}\overline{\mathcal{M}_{1}^{(-1)}}\mathcal{B}v_{1}(0) +  \mathcal{B}\overline{\mathcal{M}_{1}^{(-2)}}\mathcal{B} v_{1}'(0), \label{M1pm1p Brel} \\
& \mathcal{M}_{1}^{(0)} = \mathcal{B}\overline{\mathcal{M}_{1}^{(0)}}\mathcal{B}v_{1}(0) + \mathcal{B}\overline{\mathcal{M}_{1}^{(-1)}}\mathcal{B}v_{1}'(0)  + \mathcal{B}\overline{\mathcal{M}_{1}^{(-2)}}\mathcal{B} \frac{v_{1}''(0)}{2}. \label{M1p0p Brel} 
\end{align}
From the first column of \eqref{M1pm2p Brel}, we deduce that $\alpha \in \mathbb{R}$. The second column of \eqref{M1pm1p Brel} implies that $\delta \in \mathbb{R}$, and then the first column of \eqref{M1pm1p Brel}, together with \eqref{r1primeat0expression}, shows that $\beta \in \mathbb{R}$ and $\gamma \in \mathbb{R}$. From the third column of \eqref{M1p0p Brel}, we infer that $\epsilon \in \mathbb{R}$. 
\end{proof}

\begin{lemma}\label{unitdetlemma}
$M$ has unit determinant. 
\end{lemma}
\begin{proof}
Fix $(x,t) \in \mathcal{N}$. The determinant $\det M(x,t,\cdot)$ is analytic in $\C \setminus \{0\}$ and approaches $1$ as $k \to \infty$. The behavior (\ref{singRHMat0}) of $M_{1}(x,t,k)$ as $k \to 0$ implies by a direct computation that $\det M$ has at most a simple pole at $k = 0$. Thus,
$$\det M(x,t,k) = 1 + \frac{f(x,t)}{k},$$
for some function $f(x,t)$. On the other hand, by (\ref{singMatinfty}), $\det M(x,t,k) = 1+O(k^{-3})$ as $k \to \infty$ and thus $f(x,t) = 0$.
\end{proof}

\begin{lemma}
The solution $M$ of RH problem \ref{RHgoodboussinesq} is unique. 
\end{lemma}
\begin{proof}
Suppose $M$ and $N$ are two solutions. By Lemma \ref{unitdetlemma}, $\det M$ and $\det N$ are identically equal to one. In particular, the inverse transpose $N^{A}:=(N^{-1})^T$ of the matrix $N$ can be expressed in terms of its minors. Expanding this expression for $N^A$ as $k \to 0$ in $D_1$ and using (\ref{singRHMat0}), we find, as $k \in D_1$ approaches $0$,
\begin{align}\label{N1Aat0}
N_1^A(x,t,k) = 
\frac{1}{k^2} 
\begin{pmatrix}
0 & 0 & *  \\
0 & 0 & * \\
0 & 0 & *
\end{pmatrix} 
+ \frac{1}{k} \begin{pmatrix}
0  & * & * \\
0  & * & * \\
0  & * & *
\end{pmatrix} 
+  O(1).
\end{align}
Similarly, expanding the expression for $N^A$ as $k \to \infty$ and using (\ref{singMatinfty}), we find 
\begin{align}\nonumber
N^A(x,t,k) = &\; I - 
\frac{N_{33}^{(1)}(x,t)}{k} 
\begin{pmatrix} \omega^2 & 0 & 0 \\ 
0 & \omega & 0 \\ 
0 & 0 & 1
\end{pmatrix}
+ \frac{1}{k^2}\Bigg\{(N_{33}^{(1)}(x,t)^2 - N_{33}^{(2)}(x,t))
\begin{pmatrix} \omega & 0 & 0 \\ 
0 & \omega^2 & 0 \\ 
0 & 0 & 1
\end{pmatrix}
	\\ \label{N1Aatinfty}
& + \frac{\widetilde{N}^{(2)}_{13}(x,t)}{1-\omega} \begin{pmatrix} 0 & \omega & -\omega^2 \\ 
-1 & 0 & \omega^2  \\ 
1 & -\omega & 0
\end{pmatrix}\Bigg\}
+ O(k^{-3}), \qquad k\to \infty,
\end{align}
where $\widetilde{N}^{(2)}_{13} = -(1-\omega)N^{(2)}_{13}$.
By (\ref{singRHMat0}) and (\ref{N1Aat0}), we have, as $k \in \bar{D}_1$ approaches $0$,
\begin{align}\nonumber
MN^{-1} = &\; \Bigg\{\frac{\alpha}{k^2} 
\begin{pmatrix}
\omega & 0 & 0 \\
\omega & 0 & 0 \\
\omega & 0 & 0
\end{pmatrix} 
+ \frac{1}{k} \begin{pmatrix}
*  & * & 0 \\
*  & * & 0 \\
*  & * & 0 
\end{pmatrix} +  \begin{pmatrix}
* & * & \epsilon \\
* & * & \epsilon \\
* & * & \epsilon 
\end{pmatrix} +  O(k)\Bigg\}
	\\ \label{MNinvat0}
& \times \Bigg\{\frac{1}{k^2} 
\begin{pmatrix}
0 & 0 & 0  \\
0 & 0 & 0\\
* & * & *
\end{pmatrix} 
+ \frac{1}{k} \begin{pmatrix}
0  & 0 & 0 \\
*  & * & * \\
*  & * & *
\end{pmatrix} 
+  O(1)\Bigg\}
= O(k^{-2}),
\end{align}
showing that $MN^{-1}$ has at most a double pole at $k = 0$.
Since $MN^{-1}$ is analytic for $k \in \C \setminus \{0\}$ and approaches the identity matrix as $k \to \infty$, we conclude that
\begin{align}\label{MNinvBC}
M(x,t,k)N(x,t,k)^{-1} = I + \frac{P(x,t)}{k} + \frac{Q(x,t)}{k^2}
\end{align}
for some matrices $P(x,t)$ and $Q(x,t)$. In fact, keeping track of the terms of order $k^{-2}$ in (\ref{MNinvat0}), we see that 
\begin{align}\label{Q1jQ2jQ3j}
Q_{1j} = Q_{2j} = Q_{3j}, \qquad j = 1,2,3.
\end{align}
Furthermore, the symmetries (\ref{singRHsymm}) hold for $M$ and $N$ and hence also for $MN^{-1}$. Thus, $P$ and $Q$ have the form
\begin{align*}
&P  = \begin{pmatrix} \omega^2 P_{33} & \omega \overline{P_{13}} & P_{13} \\ 
\omega^2 P_{13} & \omega P_{33} & \overline{P_{13}} \\ 
\omega^2 \overline{P_{13}} & \omega P_{13} &  P_{33}
\end{pmatrix}, 
\qquad
 Q = \begin{pmatrix} \omega Q_{33} & \omega^2 \overline{Q_{13}} & Q_{13} \\ 
\omega Q_{13} & \omega^2 Q_{33} & \overline{Q_{13}} \\ 
\omega \overline{Q_{13}} & \omega^2 Q_{13} &  Q_{33}
\end{pmatrix}, 
\end{align*}
where $P_{33}, Q_{33}$ are real-valued and $P_{13}, Q_{13}$ are complex-valued functions. Together with (\ref{Q1jQ2jQ3j}), this implies that
\begin{align}\label{QQ33pmatrix}
Q = Q_{33} \begin{pmatrix} \omega & \omega^2 & 1 \\ \omega & \omega^2 & 1 \\ \omega & \omega^2 & 1 \end{pmatrix}.
\end{align}

On the other hand, substituting the expansion (\ref{singMatinfty}) of $M$ and the transpose of the expansion (\ref{N1Aatinfty}) of $N^A$ into the left-hand side of (\ref{MNinvBC}) and identifying coefficients of $k^{-1}$ and $k^{-2}$ in the resulting equation, we conclude that
\begin{align}\label{PMexpression}
P = &\; (M^{(1)}_{33} - N^{(1)}_{33})\begin{pmatrix} \omega^2 & 0 & 0 \\ 
0 & \omega & 0 \\ 
0 & 0 & 1
\end{pmatrix},
	\\\nonumber
Q = &\; \Big(M^{(2)}_{33} - M^{(1)}_{33}N^{(1)}_{33} + (N^{(1)}_{33})^2 - N^{(2)}_{33}\Big)\begin{pmatrix} \omega & 0 & 0 \\ 
0 & \omega^2 & 0 \\ 
0 & 0 & 1
\end{pmatrix}
	\\ \label{QMexpression}
& + \frac{\widetilde{M}^{(2)}_{13} - \widetilde{N}^{(2)}_{13}}{1-\omega}\begin{pmatrix} 0 & 1 & -1 \\ 
-\omega & 0 & \omega \\ 
\omega^2 & -\omega^2 & 0
\end{pmatrix}.
\end{align}
In order to reconcile the expressions (\ref{QMexpression}) and (\ref{QQ33pmatrix}) for $Q$, we must have $Q = 0$, and then (\ref{MNinvBC}) becomes
$$M(x,t,k)N(x,t,k)^{-1} = I + \frac{M^{(1)}_{33}(x,t) - N^{(1)}_{33}(x,t)}{k}\begin{pmatrix} \omega^2 & 0 & 0 \\ 
0 & \omega & 0 \\ 
0 & 0 & 1
\end{pmatrix}.$$
Since the determinant of the left-hand side is identically equal to one, we conclude that $M^{(1)}_{33} - N^{(1)}_{33} = 0$. This shows that $M = N$ and completes the proof of the lemma.
\end{proof}

\section{Properties of the solution $m$ of RH problem \ref{RH3x3}}\label{msec}

We establish several properties of the solution $m$ of RH problem \ref{RH3x3}.
Throughout this section, we assume that $\mathcal{N}$ is a (not necessarily open) subset of $\mathbb{R}\times [0,+\infty)$ and that $m(x,t,\cdot)$ satisfies RH problem \ref{RH3x3} for each $(x,t) \in \mathcal{N}$.

\begin{lemma}\label{muniquelemma}
The solution $m$ of RH problem \ref{RH3x3} is unique. 
\end{lemma}
\begin{proof}
The determinant $\det m$ is analytic in $\C \setminus \{0\}$ with a removable singularity at $0$ and approaches $1$ as $k \to \infty$. Hence $\det m = 1$. Thus, if $m$ and $n$ are two solutions, then $n$ is invertible and $n^{-1}$ is $O(1)$ as $k \to 0$. Hence, $mn^{-1}$ is an analytic function of $k \in \C \setminus \{0\}$ with a removable singularity at $0$ such that $mn^{-1} \to I$ as $k \to \infty$. Thus $mn^{-1}=I$. 
\end{proof}

\begin{lemma}\label{mregularsymmlemma}
The solution $m$ obeys the symmetries
\begin{align}\label{regRHsymm}
m(x,t, k) = \mathcal{A} m(x,t,\omega k)\mathcal{A}^{-1} = \mathcal{B} \overline{m(x,t,\overline{k})}\mathcal{B}, \qquad k \in \mathbb{C}\setminus \Gamma.
\end{align}
\end{lemma}
\begin{proof}
Since the jump matrix $v$ obeys the symmetries 
\begin{align}\label{vsymm}
v(x,t,k) =  \mathcal{A}v(x,t,\omega k) \mathcal{A}^{-1} = \mathcal{B}\overline{v(x,t,\overline{k})}^{-1} \mathcal{B}, \qquad k \in \Gamma,
\end{align}
the functions $\mathcal{A} m(x,t,\omega k)\mathcal{A}^{-1}$ and $\mathcal{B} \overline{m(x,t,\overline{k})}\mathcal{B}$ also satisfy RH problem \ref{RHgoodboussinesq}. The lemma follows by the uniqueness established in Lemma \ref{muniquelemma}.
\end{proof}

\begin{lemma}[Asymptotics of $m$ as $k \to 0$]\label{mat0lemma}
Suppose $r_{1}$ and $r_{2}$ satisfy Assumption \ref{r1r2assumption}.
There exist formal power series $\{m_{j}^{L}(x,t,k)\}_{j=1}^{6}$ of the form
\begin{align*}
m_{j}^{L}(x,t,k) = \sum_{l=0}^{+\infty} \mathfrak{m}_{j}^{(l)}(x,t)k^{l},
\end{align*}
such that $m(x,t,k)$ coincides with $m_{j}^{L}(x,t,k)$ to any order as $k \in \bar{D}_{j}$ tends to $0$. More precisely, for any $N \geq 0$, $(x,t) \in \mathcal{N}$, and $j = 1, \dots, 6$,
\begin{align}\label{mregularat0}
& m(x,t,k) = \sum_{l=0}^{N}\mathfrak{m}_j^{(l)}(x,t)k^{l}+ O(k^{N+1}), \qquad k \to 0, \ k \in \bar{D}_j.
\end{align}
The formal power series $m_{1}^{L}$ satisfies
\begin{align}
m_{1}^{L}(x,t,k) & = \mathcal{A}m_{1}^{L}(x,t,\omega k)\mathcal{A}^{-1}v_{6}^{L}(x,t,k)v_{1}^{L}(x,t,k), \label{Acal symmetries of m1L formal power series} \\
& = \mathcal{B}\overline{m_{1}^{L}(x,t,\overline{k})}\mathcal{B}v_{1}^{L}(x,t,k), \label{Bcal symmetries of m1L formal power series}
\end{align}
where $\{v_{j}^{L}\}_{j=1}^{6}$ are defined in \eqref{def of vjL}. In particular, the first coefficients satisfy the relations
\begin{subequations}\label{m10m11symm}
\begin{align}\label{m10m11symma}
& \mathfrak{m}_1^{(0)} = \mathcal{A}\mathfrak{m}_1^{(0)}\mathcal{A}^{-1}(v_6v_1)(0),
	\\\label{m10m11symmb}
& \mathfrak{m}_1^{(1)} = \omega \mathcal{A}\mathfrak{m}_1^{(1)}\mathcal{A}^{-1}(v_6v_1)(0) + \mathcal{A}\mathfrak{m}_1^{(0)}\mathcal{A}^{-1}(v_6v_1)'(0),
	\\\label{m10m11symmc}
&   \mathfrak{m}_1^{(0)} = \mathcal{B}\overline{\mathfrak{m}_1^{(0)}}\mathcal{B}v_1(0),
	\\\label{m10m11symmd}
&   \mathfrak{m}_1^{(1)} = \mathcal{B}\overline{\mathfrak{m}_1^{(1)}}\mathcal{B}v_1(0) +\mathcal{B}\overline{\mathfrak{m}_1^{(0)}}\mathcal{B}v_1'(0).
\end{align}
\end{subequations}
If $r_{1}$ and $r_{2}$ also satisfy Assumption \ref{r1r2at0assumption}, then 
\begin{align*}
 v_1(0) =&\; \begin{pmatrix}1 & -\omega & 0 \\ \omega^2 & 0 & 0 \\ 0 &0 & 1\end{pmatrix},
\qquad (v_6v_1)(0) = \begin{pmatrix}1 & -\omega & 1 \\ 
\omega^2 & 0 & 0 \\ -1 & \omega & 0 \end{pmatrix},
	\\
v_1'(0) = &\; \begin{pmatrix}0 & (1-\omega^2)x - r_1'(0) & 0 \\ 
(\omega-1)x - \omega \, r_1'(0) & 0 & 0 \\ 0 &0 & 0 \end{pmatrix}, 
	\\
(v_6v_1)'(0) = &\; \begin{pmatrix}0 & 1-\omega^2 & \omega-1 \\ \omega - 1 & 0 & 0 \\ \omega-1 & \omega-1 & 0 \end{pmatrix}x  + 
\begin{pmatrix}
0 & -r_1'(0) & \omega^2 r_2'(0) \\
-\omega r_1'(0) & 0 & 0 \\
\omega^2 r_2'(0) & r_1'(0) - r_2'(0) & 0
\end{pmatrix}
\end{align*}
and the leading coefficient $\mathfrak{m}^{(0)}_{1}$ has the form
\begin{align}\label{m10structure}
\mathfrak{m}_1^{(0)} = \begin{pmatrix} \mathfrak{m}^{(0)}_{11} & \mathfrak{m}^{(0)}_{12} & \mathfrak{m}^{(0)}_{33} \\ 
\omega^2 \mathfrak{m}^{(0)}_{11}+\mathfrak{m}^{(0)}_{33}-\mathfrak{m}^{(0)}_{12}  & \omega(\mathfrak{m}^{(0)}_{12}-\mathfrak{m}^{(0)}_{33}) & \mathfrak{m}^{(0)}_{33} \\
\omega \mathfrak{m}^{(0)}_{11} + \mathfrak{m}^{(0)}_{12 } & \omega^2 \mathfrak{m}^{(0)}_{12} + \mathfrak{m}^{(0)}_{33} & \mathfrak{m}^{(0)}_{33}
\end{pmatrix}
\end{align}
for some real-valued functions $\mathfrak{m}_{11}^{(0)}$ and $\mathfrak{m}_{33}^{(0)}$ and some complex-valued function $\mathfrak{m}^{(0)}_{12}$ such that
\begin{align}\label{realpartm012}
2\re(\mathfrak{m}^{(0)}_{12}) = \mathfrak{m}^{(0)}_{33}.
\end{align}
\end{lemma}
\begin{proof}
For notational convenience, we omit the $(x,t)$-dependence of several quantities. 
The existence of the formal power series $m_{1}^{L}$ follows directly from condition $(d)$ of RH problem \ref{RH3x3}. The existence of the other formal series $m_{j}^{L}$, $j=2, \dots, 6$ then follows from the symmetries \eqref{regRHsymm} of $m$. Since the jumps $v_{j}$ can be Taylor expanded as $k \to 0$, a Taylor expansion of \eqref{regRHsymm} together with \eqref{vsymm} implies that
\begin{align*}
m_{1}^{L}(k) = \mathcal{A}m_{3}^{L}(\omega k) \mathcal{A}^{-1} = \mathcal{A}m_{1}^{L}(\omega k) v_{2}^{L}(\omega k) v_{3}^{L}(\omega k) \mathcal{A}^{-1} = \mathcal{A}m_{1}^{L}(\omega k) \mathcal{A}^{-1} v_{6}^{L}(k)v_{1}^{L}(k)
\end{align*}
and
\begin{align*}
m_1^{L}(k) = \mathcal{B}\overline{m_6^{L}(\overline{k})}\mathcal{B} = \mathcal{B}\overline{m_1^{L}(\overline{k})}\overline{v_1^{L}(\overline{k})^{-1}}\mathcal{B}
= \mathcal{B}\overline{m_1^{L}}(\overline{k})\mathcal{B}v_1^{L}(k),
\end{align*}
which proves \eqref{Acal symmetries of m1L formal power series} and \eqref{Bcal symmetries of m1L formal power series}, and in particular the relations (\ref{m10m11symm}).

If $r_{1}$ and $r_{2}$ satisfy Assumption \ref{r1r2at0assumption}, then we find as in the proof of Lemma \ref{realitylemma} that the relations \eqref{identity for r1'}--\eqref{identity for r1''r2''} hold.
Using these relations, we can express $\overline{r_{1}'(0)}$, $\overline{r_{2}'(0)}$, and $\overline{r_{2}''(0)}$ in terms of $r_{1}'(0)$, $r_{2}'(0)$, $r_{1}''(0)$, $\overline{r_{1}''(0)}$, and $r_{2}''(0)$. This leads to the asserted expressions for $v_1(0)$, $v_1'(0)$, $(v_6v_1)(0)$, and $(v_6v_1)'(0)$. Using the explicit expression for $(v_6v_1)(0)$, a computation shows that (\ref{m10m11symma}) implies that $\mathfrak{m}_1^{(0)}$ has the form (\ref{m10structure}) for some complex-valued functions $\mathfrak{m}^{(0)}_{11}$, $\mathfrak{m}^{(0)}_{12}$, and $\mathfrak{m}^{(0)}_{33}$.
The relation (\ref{m10m11symmc}) then implies that $\mathfrak{m}_{11}^{(0)}$ and $\mathfrak{m}_{33}^{(0)}$ are real-valued and that $\mathfrak{m}^{(0)}_{12}$ satisfies (\ref{realpartm012}).
\end{proof}

\begin{lemma}\label{matinftylemma}
Suppose $r_{1}$ and $r_{2}$ satisfy Assumption \ref{r1r2assumption}.
For each $(x,t) \in \mathcal{N}$, there exists an asymptotic expansion
\begin{align}\label{mexpansionatinfty}
 m(x,t,k) = I + \frac{m^{(1)}(x,t)}{k} + \frac{m^{(2)}(x,t)}{k^2} + O(k^{-3}), \qquad k \to \infty,
\end{align}
where the coefficient matrices $\{m^{(j)}\}_1^2$ have the form
\begin{align}\label{m1m2coefficients}
m^{(1)} = \begin{pmatrix} \omega^2 m_{33}^{(1)} & \omega \overline{m_{13}^{(1)}} & m_{13}^{(1)} \\ 
\omega^2 m_{13}^{(1)} & \omega m_{33}^{(1)} & \overline{m_{13}^{(1)}} \\ 
\omega^2 \overline{m_{13}^{(1)}} & \omega m_{13}^{(1)} &  m_{33}^{(1)}
\end{pmatrix}, 
\qquad
m^{(2)} = \begin{pmatrix} \omega m_{33}^{(2)} & \omega^2 \overline{m_{13}^{(2)}} & m_{13}^{(2)} \\ 
\omega m_{13}^{(2)} & \omega^2 m_{33}^{(2)} & \overline{m_{13}^{(2)}} \\ 
\omega \overline{m_{13}^{(2)}} & \omega^2 m_{13}^{(2)} &  m_{33}^{(2)}
\end{pmatrix}
\end{align}
for some real-valued functions $\{m_{33}^{(j)}\}_1^2$ and complex-valued functions $\{m_{13}^{(j)}\}_1^2$. 
\end{lemma}
\begin{proof}
The existence of an expansion of the form \eqref{mexpansionatinfty} follows from Cauchy's formula
\begin{align}\label{mrepresentation}
m(x,t,k) = I + \frac{1}{2\pi i}\int_{\Gamma}\frac{m_{-}(x,t,s)(v(x,t,s)-I)}{s-k}ds,
\end{align}
together with the fact that $r_{1}(k)$ and $r_{2}(k)$ are smooth, and have rapid decay as $k \to +\infty$ and as $k \to -\infty$, respectively (see Assumption \ref{r1r2assumption}).
Formula \eqref{m1m2coefficients} is then a consequence of the symmetries established in Lemma \ref{mregularsymmlemma}.
\end{proof}

\section{Proof of Theorem \ref{mainth}}\label{proofsec}
Suppose that $r_{1}(k)$ and $r_{2}(k)$ satisfy Assumptions \ref{r1r2assumption} and \ref{r1r2at0assumption} and define $v(x,t,k)$ by \eqref{vdef}. Assume that $m(x,t,k)$ solves RH problem \ref{RH3x3} for all $(x,t) \in \mathcal{N}$, where $\mathcal{N}$ is an open subset of $\mathbb{R}\times [0,+\infty)$. 

By Lemma \ref{muniquelemma}, the solution $m$ is unique.
By Lemma \ref{matinftylemma}, the function $q(x,t)$ is well-defined by (\ref{recoverq3x3}). 
Viewing (\ref{mrepresentation}) as a singular integral equation for $m_-$ and using that $r_1(k)$ and $r_2(k)$ have rapid decay as $k \to \infty$, standard arguments imply that $m(x,t,k)$ and $q(x,t)$ are smooth functions of $(x,t) \in \mathcal{N}$, see e.g. \cite[Section 5]{LNonlinearFourier} for details in a similar situation. Assertion $(\ref{partagood})$ of Theorem \ref{mainth} then follows from the following lemma.

\begin{lemma}\label{lemma: lax pair for m}
Define $q(x,t)$ by (\ref{recoverq3x3}) for $(x,t) \in \mathcal{N}$. Then $m$ satisfies the Lax pair equations
\begin{align}\label{mlax}
\begin{cases}
m_x-k[J,m]=\mathcal{U}m,\\
m_t-k^2[J^2,m]=\mathcal{V}m.
\end{cases}
 \qquad (x,t) \in \mathcal{N}, \ k \in \C \setminus \Gamma,
\end{align}
where $\mathcal{U}$  and $\mathcal{V}$ are defined in terms of $q$ by \eqref{expression for Ufrak} and \eqref{expression for Vfrak}, and $J = \mbox{\upshape diag}(\omega, \omega^{2}, 1)$. In particular, $q$ satisfies (\ref{3x3eq}) for $(x,t) \in \mathcal{N}$.
\end{lemma}
\begin{proof}
Since $\det m = 1$, the inverse $m^{-1}$ is well defined.
Recalling the definition \eqref{vdef} of the jump matrix $v$, we can write
\begin{align*}
v(x,t,k) = e^{kJx + k^{2}J^{2}t}v_0(k)e^{-kJx - k^{2}J^{2}t}
\end{align*}
with a matrix $v_0(k)$ independent of $x$ and $t$. This yields
\begin{align*}
v_{x} = k [J,v] \qquad \mbox{ and } \qquad v_{t} = k^{2}[J^{2},v],
\end{align*}
from which we deduce 
\begin{align*}
\begin{cases}
 m_{x,+} = m_{x,-}v + km_{-}[J,v], \\
 m_{t,+} = m_{t,-}v + k^{2}m_{-}[J^{2},v],
\end{cases} \quad k \in \Gamma,
\end{align*}
or, equivalently,
\begin{align*}
\begin{cases}
m_{x,+}+km_{+}J = (m_{x,-}+km_{-}J)v, \\
m_{t,+}+k^{2}m_{+}J = (m_{t,-}+k^{2}m_{-}J)v,
\end{cases} \quad k \in \Gamma.
\end{align*}
Therefore, the functions
$$(m_x + kmJ)m^{-1} \qquad \text{and} \quad (m_t + k^2mJ^2)m^{-1}$$
are analytic for $k \in \C \setminus \{0\}$. As $k \to \infty$, they are $O(k)$ and $O(k^2)$, respectively, and as $k \to 0$ they are $O(1)$.
Hence there exist $3\times 3$-matrix valued functions $\{F_j(x,t)\}_{j=0}^1$ and $\{G_j(x,t)\}_{j=0}^2$ such that
\begin{align}\label{mlaxFG}
\begin{cases}
m_x + kmJ = (kF_1 + F_0)m, 
	\\
m_t + k^2mJ^2 = (k^2G_2 + kG_1 + G_0)m.
\end{cases}
\end{align}
The Lax pair equations (\ref{mlax}) follow from (\ref{mlaxFG}) if we can show that
\begin{subequations}\label{3x3FGUV}
\begin{align}
& F_1 = J, && F_0 = \mathcal{U},
	\\
& G_2 = J^2, && kG_1 + G_0 = \mathcal{V},
\end{align}
\end{subequations}
where $\mathcal{U}$ and $\mathcal{V}$ are given by (\ref{expression for Ufrak}) and \eqref{expression for Vfrak}, respectively.

To prove (\ref{3x3FGUV}), we note that the terms of $O(k^2)$ as $k \to \infty$ of (\ref{mlaxFG}) show that $G_2 = J^2$, and the terms of $O(k)$ and $O(1)$ then yield 
\begin{subequations}\label{FGintermsofm}
\begin{align}
F_1 = J, \qquad G_1 = [m^{(1)}, J^2],
\end{align}
and
\begin{align}
F_0 = [m^{(1)}, J], \qquad G_0 = [m^{(2)}, J^2] - [m^{(1)}, J^2]m^{(1)}.
\end{align}
\end{subequations}
An explicit computation of these equations gives $F_{0} = \mathcal{U}$ and
\begin{align}
& G_{1} = \begin{pmatrix} 
0 & (\omega^2-1) \bar{q} & (1-\omega^2)q \\
(\omega-1)q & 0 & (1-\omega)\bar{q} \\
(\omega-\omega^2)\bar{q} & (\omega^2-\omega)q & 0 \end{pmatrix}, \\
& G_{0} = \begin{pmatrix}
0 & (1-\omega)\bar{\frak{u}} & (1-\omega^{2})\frak{u} \\
(1-\omega^{2})\frak{u} & 0 & (1-\omega)\bar{\frak{u}} \\
(1-\omega)\bar{\frak{u}} & (1-\omega^{2})\frak{u} & 0
\end{pmatrix}, \label{explicit G1 G0}
\end{align}
where $\frak{u} = m_{13}^{(2)}-m_{33}^{(1)}q + \bar{q}^{2}$. After substituting these expressions for $\{F_j\}_0^1$ and $\{G_j\}_0^2$ into the compatibility condition $m_{xt} = m_{tx}$ and simplifying, we find the relation
\begin{multline}\label{compatibility explicit}
kF_{1,t}+F_{0,t}+(kF_{1}+F_{0})(k^{2}G_{2}+kG_{1}+G_{0}) \\ = k^{2}G_{2,x}+kG_{1,x}+G_{0,x}+(k^{2}G_{2}+kG_{1}+G_{0})(kF_{1}+F_{0}).
\end{multline}
The terms of order $k$ in \eqref{compatibility explicit} yield
\begin{align}\label{compatibility order k}
F_{1}G_{0}+F_{0}G_{1} = G_{1,x}+G_{1}F_{0}+G_{0}F_{1}.
\end{align}
Using the explicit formulas for $F_{0}$, $F_{1}$, and $G_{1}$, we find that \eqref{compatibility order k} is equivalent to
\begin{align}\label{get rid of u}
\frak{u} = \frac{1-\omega}{3\omega}(q_{x}-3\bar{q}^{2}).
\end{align}
The relations (\ref{3x3FGUV}) follow by substituting (\ref{get rid of u}) into (\ref{explicit G1 G0}). 

The terms of order $1$ in \eqref{compatibility explicit} yield
\begin{align*}
F_{0,t}+F_{0}G_{0}=G_{0,x}+G_{0}F_{0},
\end{align*}
or, equivalently,
\begin{align}\label{lol4}
q_{t} = (1+\omega)\mathfrak{u}_{x} = \frac{1-\omega^{2}}{3\omega}\big( q_{xx}-6 \bar{q} \bar{q}_{x} \big)
\end{align}
where we have used \eqref{get rid of u} in the last step. Since $\omega = e^{\frac{2\pi i}{3}}$, it follows that $q$ satisfies (\ref{3x3eq}). 
\end{proof}

Assertion $(\ref{partbgood})$ follows from Lemma \ref{matinftylemma} and (\ref{recoverq3x3}).

Let $y_{1}$, $y_{2}$, and $y_{3}$ be the real-valued functions defined in terms of $q$ and $m$ by \eqref{yjdef}. Let 
\begin{align}\label{Tdef}
T(x,t,k) = I + \frac{A(x,t)}{k} + \frac{B(x,t)}{k^2}, \qquad k \in \mathbb{C}\setminus \{0\}, \quad (x,t) \in \mathcal{N},
\end{align}
where $A(x,t)$ and $B(x,t)$ are given in terms of $y_{1}$, $y_{2}$, and $y_{3}$ by (\ref{BAdef}), and define $M$ by
\begin{align}\label{def of M in terms of m}
M(x,t,k) = T(x,t,k)m(x,t,k).
\end{align}
Assertion $(\ref{partcgood})$ is a consequence of the next lemma.

\begin{lemma}
$M(x,t,k)$ defined by \eqref{def of M in terms of m} satisfies RH problem \ref{RHgoodboussinesq} for each $(x,t) \in \mathcal{N}$.
\end{lemma}
\begin{proof}
Since $M = Tm$, the function $M$ is analytic for $k \in \C \setminus \Gamma$ and satisfies the same jump conditions as $m$ on $\Gamma$. 
Moreover, since $y_{1}$, $y_{2}$, and $y_{3}$ are real-valued functions, it is easy to check that the factor $T$ has unit determinant and obeys the symmetries
$$T(x,t,k) = \mathcal{A} T(x,t,\omega k) \mathcal{A}^{-1} = \mathcal{B} \overline{T(x,t,\bar{k})} \mathcal{B}$$
for $(x,t) \in \mathcal{N}$. Since $m$ obeys the same symmetries by Lemma \ref{mregularsymmlemma}, $M$ does as well. 

Let us prove that $M$ satisfies (\ref{singRHMat0}) as $k \to 0$ with coefficients fulfilling \eqref{explicitMcalpm2p}--\eqref{explicit Mcalp0p third column} for some choice of the functions $\alpha, \beta, \gamma, \delta, \epsilon$.
We deduce from Lemma \ref{mat0lemma} and the definition (\ref{def of M in terms of m}) of $M$ that $M$ has an expansion of the form (\ref{singRHMat0}) as $k \in \bar{D}_1$ approaches $0$. The first few coefficients are given by
\begin{subequations}\label{MBAm}
\begin{align}
&\mathcal{M}_1^{(-2)} = B\mathfrak{m}_1^{(0)},
	\\
&\mathcal{M}_1^{(-1)} = B\mathfrak{m}_1^{(1)} + A\mathfrak{m}_1^{(0)},
	\\
&\mathcal{M}_1^{(0)} = B\mathfrak{m}_1^{(2)} + A\mathfrak{m}_1^{(1)} + \mathfrak{m}_1^{(0)},
\end{align}
\end{subequations}
where $\{\mathfrak{m}_1^{(l)}\}_{l=0}^{2}$ are the coefficient in (\ref{mregularat0}). 
We substitute the expression (\ref{m10structure}) for $\mathfrak{m}_1^{(0)}$ into (\ref{m10m11symmb}) and solve the six entries in the first two rows of this equation for $\mathfrak{m}^{(1)}_{31}$, $\mathfrak{m}^{(1)}_{32}$, $\mathfrak{m}^{(1)}_{13}$, $\mathfrak{m}^{(1)}_{21}$, $\mathfrak{m}^{(1)}_{22}$, and $\mathfrak{m}^{(1)}_{23}$. Substituting the resulting expressions as well as the expression (\ref{m10structure}) for $\mathfrak{m}_1^{(0)}$ into (\ref{MBAm}) and using that $r_2'(0) \in i\R$, long computations show that 
$\mathfrak{M}_1^{(-2)}$, $\mathfrak{M}_1^{(-1)}$, and $\mathfrak{M}_1^{(0)}$ have the structure specified in \eqref{explicitMcalpm2p}--\eqref{explicit Mcalp0p third column},
where $\alpha, \beta, \gamma, \delta, \epsilon$ are given explicitly in terms of $r_1'(0)$, $r_2'(0)$, $\{y_j\}_1^3$, and the entries of $m_1^{(0)}$ and $m_1^{(1)}$. This shows that $M$ has the desired behavior as $k \to 0$.

We finally show that the asymptotics \eqref{Matinfty} hold as $k \to \infty$ with matrix-coefficients $M^{(1)}$ and $M^{(2)}$ satisfying (\ref{singRHMatinftyb}). As $k \to \infty$, $m$ obeys the expansion (\ref{mexpansionatinfty}). Hence $M$ satisfies
$$M = I + \frac{M^{(1)}}{k} + \frac{M^{(2)}}{k^2} + O(k^{-3}), \qquad k\to \infty,$$
where 
\begin{align}\label{M1M2}
M^{(1)} = A + m^{(1)}, \qquad M^{(2)} = B + A m^{(1)} + m^{(2)}.
\end{align}
Using that $\{y_j\}_1^3$ are given by (\ref{yjdef}), we conclude that the three conditions in (\ref{singRHMatinftyb}) are satisfied. This completes the proof.
\end{proof}

Using (\ref{M1M2}), we can write the definition (\ref{recoverugood}) of $u(x,t)$ as
\begin{align}\label{recoveruAm}
 u(x,t) 
 = -\frac{3}{2}\frac{\partial}{\partial x} M_{33}^{(1)}(x,t) = -\frac{3}{2}\frac{\partial}{\partial x} \big(A_{33}(x,t) + m_{33}^{(1)}(x,t)\big),
\end{align}
By considering the terms of order $k^{-1}$ of the $(33)$-entry of the first equation in (\ref{mlax}), we infer that $m_{33}^{(1)}$ satisfies 
\begin{align}\label{m133x}
\frac{\partial}{\partial x} m_{33}^{(1)}(x,t) = 3|q(x,t)|^2, \qquad (x,t) \in \mathcal{N}.
\end{align}
Substituting this into (\ref{recoveruAm}) and using the definition (\ref{BAdef}) of $A(x,t)$, we obtain
\begin{align*}
 u(x,t) 
 = -\frac{3}{2}\frac{\partial}{\partial x} \bigg(\frac{\omega q -  \bar{q}}{1-\omega} + \frac{\omega(q-\bar{q})}{1-\omega^2}\bigg)
 -\frac{9}{2}|q(x,t)|^2
 = - 3 \re\big(\omega q_x(x,t)\big)  -\frac{9}{2}|q(x,t)|^2 ,
\end{align*}
which is the Miura transformation (\ref{miuragood}). Since $q$ is smooth, it follows from (\ref{miuragood}) that $u$ is smooth and real-valued on $\mathcal{N}$. Since $q$ solves (\ref{3x3eq}), it also follows from (\ref{miuragood}) that $u$ satisfies the ``good'' Boussinesq equation (\ref{goodboussinesqnouxx}) and that $\{u, v\}$, with $v$ defined by (\ref{vintermsofq}), satisfy the system (\ref{boussinesqsystem}) on $\mathcal{N}$.
For completeness, we give in the next lemma an alternative proof of the latter two facts; this lemma also shows that $M$ satisfies the Lax pair equations associated to (\ref{boussinesqsystem}).

\begin{lemma}
For $(x,t) \in \mathcal{N}$, define $u(x,t)$ by (\ref{recoverugood}) and define $v(x,t)$ by (\ref{vintermsofq}).
Then $M$ satisfies the Lax pair equations 
\begin{align}\label{Msinglax}
\begin{cases}
M_x - k[J,M] = \mathcal{U} M, 
	\\
M_t - k^2 [J^2,M] = \mathcal{V} M,
\end{cases}
\end{align}
where $\mathcal{U}$ and $\mathcal{V}$ are given in terms of $u$ and $v$ by (\ref{UVexpressions}). In particular, for $(x,t) \in \mathcal{N}$, $u$ and $v$ satisfy (\ref{boussinesqsystem}) and $u$ satisfies (\ref{goodboussinesqnouxx}).
\end{lemma}
\begin{proof}
Since $T$ and $m$ have unit determinant, so does $M$; hence the inverse $M^{-1}$ is well defined.
Since $M = Tm$, we have
\begin{subequations}\label{MMinvT}
\begin{align}\label{MMinvTa}
(M_x + kMJ)M^{-1} = T_xT^{-1} + T(kF_1 + F_0)T^{-1}
\end{align}
and
\begin{align}\label{MMinvTb}
(M_t + k^2MJ^2)M^{-1} = T_{t}T^{-1} + T(k^2G_2 + k G_1 + G_0)T^{-1},
\end{align}
\end{subequations}
where $\{F_j\}_{j=0}^1$ and $\{G_j\}_{j=0}^2$ are the functions in (\ref{mlaxFG}).
It follows from the definition of $T$ that the functions in (\ref{MMinvTa}) and (\ref{MMinvTb}) are analytic for $k \in \C \setminus \{0\}$. Moreover,
\begin{align*}
& T_xT^{-1} + T(kF_1 + F_0)T^{-1} = kJ+O(1) & & \mbox{as } k \to \infty, \\
& T_{t}T^{-1} + T(k^2G_2 + k G_1 + G_0)T^{-1} = k^{2}J^{2}+O(k) & & \mbox{as } k \to \infty,
\end{align*}
and we also verify from a direct computation that
\begin{align*}
& T_xT^{-1} + T(kF_1 + F_0)T^{-1} = O(k^{-2}) & & \mbox{as } k \to 0, \\
& T_{t}T^{-1} + T(k^2G_2 + k G_1 + G_0)T^{-1} = O(k^{-2}) & & \mbox{as } k \to 0.
\end{align*}
Hence there exist $3\times 3$-matrix valued functions $\{f_j(x,t)\}_{j=-2}^0$ and $\{g_j(x,t)\}_{j=-2}^1$ such that
\begin{align}\label{laxFG}
\begin{cases}
T_xT^{-1} + T(kF_1 + F_0)T^{-1} = kJ + f_0 + \frac{f_{-1}}{k} + \frac{f_{-2}}{k^2}, 
	\\
T_xT^{-1} + T(k^2G_2 + k G_1 + G_0)T^{-1} = k^2J^{2} + kg_1 + g_0 + \frac{g_{-1}}{k} + \frac{g_{-2}}{k^2}.
\end{cases}
\end{align}
The Lax pair equations (\ref{Msinglax}) are a consequence of (\ref{MMinvT}) and (\ref{laxFG}) if we can show that
\begin{align}\label{FGUV}
& f_0 + \frac{f_{-1}}{k} + \frac{f_{-2}}{k^2} = \mathcal{U},
	& & kg_1 + g_0 + \frac{g_{-1}}{k} + \frac{g_{-2}}{k^2} = \mathcal{V},
\end{align}
where $\mathcal{U}$ and $\mathcal{V}$ are given by (\ref{UVexpressions}).
To prove (\ref{FGUV}), we substitute (\ref{3x3FGUV}) and the definition (\ref{Tdef}) of $T$ into (\ref{laxFG}) and identify coefficients of powers of $k$. Recalling the definitions of $u$ and $v$, the identities in (\ref{FGUV}) follow from long but straightforward computations in which we use the relation \eqref{get rid of u}. 
The facts that $u$ and $v$ satisfy (\ref{boussinesqsystem}) and that $u$ satisfies (\ref{goodboussinesqnouxx}) follow from the compatibility of the Lax pair equations (\ref{Msinglax}).
\end{proof}

The proof of Theorem \ref{mainth} is complete.

\section{Examples}\label{section:example}
We give two explicit examples to illustrate the transformation \eqref{miuragoodintro} of Theorem \ref{miurath}.

\subsection{Example 1}
Equation (\ref{3x3eq}) admits the family of traveling wave solutions
\begin{align}\label{qonesoliton}
q(x,t) = -\frac{c \sqrt{3}\, \omega^2 e^{\frac{i \sqrt{3}}{2} c (x -c t)}}{2 \sin\big(\frac{3 \sqrt{3}}{2} c (x-c t)\big)},
\end{align}
parametrized by the wave speed $c \in \R$. The solution in (\ref{qonesoliton}) is smooth in the $xt$-plane except on the lines $\frac{3 \sqrt{3}}{2} c (x-c t) \in \pi \Z$. Applying Theorem \ref{miurath}, we find that, for any $c_1 = 3^{1/4}c \in \R$,
$$u(x,t) = \frac{1}{2} \left(1- \frac{3 c_1^2}{\sin^2\left(\frac{1}{2} c_1 (x - c_1 t)\right)}\right)$$
is a solution of the ``good'' Boussinesq equation (\ref{goodboussinesq}) everywhere in the $xt$-plane except on the lines $ \frac{1}{2} c_1 (x - c_1 t) \in \pi \Z$.

The solution (\ref{qonesoliton}) is a (singular) one-soliton of (\ref{3x3eq}) and we explain in Appendix \ref{solitonapp} how it can be derived.

\subsection{Example 2}\label{example2sec}
Equation (\ref{3x3eq}) admits the exact solution
\begin{align}\label{qbreather}
q(x,t) = \frac{4 (4 e^{2 x}-1) e^{-2 x-\frac{i t}{\sqrt{3}}} \big(1 +4 e^{2 x} -4 e^{x+i \sqrt{3} t}\big)}{128 \cos (\sqrt{3} t)-108 \sinh (x)+63 \sinh (3 x)-180 \cosh (x)+65 \cosh (3 x)}.
\end{align}
This solution is smooth in the $xt$-plane except when the denominator vanishes; this happens when 
$$t =  \frac{2 \pi  N \pm \arccos\left(\frac{1}{128} (108 \sinh (x)-63 \sinh (3 x)+180 \cosh (x)-65 \cosh(3 x))\right)}{\sqrt{3}}$$
for some $N \in \Z$.
An application of the Miura-type transformation of Theorem \ref{miurath} to (\ref{qbreather}) yields the following (singular) solution of (\ref{goodboussinesq}):
\begin{align*}
  u(x,t) = \frac{2 -32 \left(1+2 \sqrt{3}\right) \left(y + 4 y^3\right) \cos(\frac{\sqrt{3} t+\pi}{3})+16 y^2 \big(8 \sqrt{3}+5 -4 \sin(\frac{4 \sqrt{3} t+\pi }{6})\big)+32 y^4}{4 \Big(4 \sqrt{3} y \sin(\frac{t}{\sqrt{3}})-4 y \cos(\frac{t}{\sqrt{3}})+4 y^2+1\Big)^2},
\end{align*}
where $y := \exp(3^{-1/4} x)$. This solution is smooth whenever the denominator is nonzero.

The solution (\ref{qbreather}) is a (singular) breather soliton of (\ref{3x3eq}); we show in Appendix \ref{solitonapp} how to construct such solitons and how we arrived at the solution in (\ref{qbreather}).

\section{KdV and mKdV}\label{kdvsec}
The Miura transformation (\ref{originalmiura}) maps solutions of the mKdV equation to solutions of the KdV equation. Underlying this transformation is a correspondence between the associated Riemann--Hilbert problems. In this section, we describe this correspondence; we refer to \cite[Proposition 2.3 and Remark 2.5]{FI1994} for an earlier description.

Our goal is to give a detailed treatment of the mKdV/KdV correspondence with a particular emphasis on the analogy between that correspondence and the correspondence between the ``good'' Boussinesq equation and equation (\ref{3x3eq}) presented in Sections \ref{mainsec}--\ref{proofsec}.
To make this analogy more transparent, in this section, we write $M$ and $m$ for the solutions of the RH problems associated with the KdV and mKdV equations, respectively (in the rest of this paper, $M$ and $m$ denote the solutions to the RH problems associated with the Boussinesq equation and equation \eqref{3x3eq}, respectively). Similarly, the quantities $\mathcal{U}$, $\mathcal{V}$, $\mathcal{A}$, $\mathcal{B}$, $s$, etc., are defined differently in this section compared to earlier sections; but they play analogous roles here, so that the comparison between this section and the rest of the paper becomes evident.


\subsection{Reflection coefficients}
We first define the reflection coefficients $r_{\mKdV}$ and $r_{\KdV}$ associated with mKdV and KdV. Let $\{\sigma_j\}_{j=1}^3$ be the three Pauli matrices.
The mKdV and KdV equations \eqref{mKdV} and \eqref{KdV} both admit Lax pairs of the same form:
\begin{align}
  \begin{cases}
X_x - [\mathcal{L}, X] = \mathcal{U} X, 
  	\\
X_t - [\mathcal{Z}, X] = \mathcal{V} X,
  \end{cases} \label{mKdV Lax pair}
\end{align}
where $\mathcal{L} = -ik\sigma_{3}$, $\mathcal{Z}=-4ik^{3}\sigma_{3}$, and $\mathcal{U}(x,t,k)$ and $\mathcal{V}(x,t,k)$ are given by
\begin{align}\label{UVdefmKdVKdV}
& \mathcal{U} = \begin{cases} -q\sigma_2 = i \begin{pmatrix} 0 & q \\ -q & 0 \end{pmatrix} & (\mKdV), 
	\\
\displaystyle \frac{u}{2k}(i\sigma_3 - \sigma_1)
= \frac{u}{2k}\begin{pmatrix}
i & -1 \\ -1 & -i
\end{pmatrix} & (\KdV),
\end{cases}
	\\ \nonumber
& \mathcal{V} = \begin{cases}
\begin{pmatrix} -2ikq^2 & 4ik^2q + 2iq^3-2kq_x-iq_{xx} \\
-4ik^2q - 2iq^3 -2kq_x + iq_{xx} & 2ikq^2 \end{pmatrix} & (\mKdV), \\
\displaystyle -2ku \begin{pmatrix}
0 & 1 \\ 1 & 0
\end{pmatrix} + u_{x} \begin{pmatrix}
0 & -i \\ i & 0
\end{pmatrix} + \frac{2u^{2}-u_{xx}}{2k}\begin{pmatrix}
i & -1 \\ -1 & -i
\end{pmatrix} & (\KdV).
\end{cases}
\end{align}
Define the $2 \times 2$-matrix valued function $X(x,k)$ as the unique solution of the Volterra integral equation
\begin{align} 
 & X(x,k) = I - \int_x^{\infty} e^{(x-x')\widehat{\mathcal{L}(k)}} \mathcal{U}(x',0,k) X(x',k) dx', 
\end{align}
and define the scattering matrix $s(k)$ by 
\begin{align*}
& s(k) = I - \int_\R e^{-x\widehat{\mathcal{L}(k)}}\mathcal{U}(x,0,k) X(x,k)dx,
\end{align*}
where $e^{\widehat{\mathcal{L}}}\mathcal{U}X := e^\mathcal{L} \mathcal{U}X e^{-\mathcal{L}}$.
The mKdV reflection coefficient $r_{\mKdV}(k)$ associated to the initial data $q_0(x)$ is defined by
\begin{align}\label{mKdVrdef}
r_{\mKdV}(k) = -\frac{s_{21}(k)}{s_{22}(k)}, \qquad  k \in \mathbb{R},
\end{align}	
where $X$ and $s$  are constructed with $\mathcal{U} = -q_0\sigma_2$; the KdV reflection coefficient $r_{\KdV}(k)$ associated to the initial data $u_0(x)$ is defined by the same formula (\ref{mKdVrdef}) except that  $X$ and $s$  now are constructed with $\mathcal{U} = \frac{u_0}{2k}(i\sigma_3 - \sigma_1)$.

\subsection{Solution of the initial value problem}
In what follows, we state the RH problems associated to mKdV and KdV and recall without proofs how they can be used to solve the initial value problem on the line. 
For simplicity, we assume that no solitons are present; this means more precisely that we assume that $(s(k))_{22}$ is nonzero for $\im k \geq 0$ for both mKdV and KdV.

Given solitonless initial data $q_0$ in the Schwartz class $\mathcal{S}(\R)$, the solution $q(x,t)$ of the initial value problem for the mKdV equation is given by
\begin{align}\label{mKdVrecoverq}
q(x,t) = 2  \lim_{k\to \infty} \big(k \, m(x,t,k)\big)_{12},
\end{align}
where $m$ is the unique solution of the following RH problem with  $r(k)$ taken to be the reflection coefficient $r_{\mKdV}(k)$ associated to $q_0(x) = q(x,0)$.

\begin{RHproblem}[RH problem for mKdV]\label{RHmKdV} 
Let $r: \mathbb{R}\to \mathbb{C}$ be a given function. 
Find a $2 \times 2$-matrix valued function $m(x,t,k)$ with the following properties:
\begin{enumerate}[$(a)$]
\item $m(x,t,\cdot) : \mathbb{C}\setminus \mathbb{R} \to \mathbb{C}^{2 \times 2}$ is analytic.

\item The boundary values of $m(x,t,k)$ as $k$ approaches $\mathbb{R}$ from the left and right exist, are continuous, and satisfy
\begin{align}\label{jumpmKdV}
& m_{+}(x,t,k) = m_{-}(x,t,k)v(x,t,k), \qquad k \in \mathbb{R},
\end{align}
where $v$ is defined in terms of $r(k)$ by 
\begin{align}\nonumber
& v(x,t,k) := \begin{pmatrix} 1 - |r(k)|^2 & -\overline{r(k)}e^{-t\Phi(\zeta,k)} \\ r(k)e^{t \Phi(\zeta,k)} & 1 \end{pmatrix},  
	\\ \label{vdefmKdV}
& \Phi(\zeta, k) := 2i(\zeta k + 4k^3), \qquad \zeta := \frac{x}{t}. 
\end{align}

\item As $k \to \infty$, $m(x,t,k) = I + O(k^{-1})$.

\end{enumerate}
\end{RHproblem}

Similarly, for solitonless initial data $u_0 \in \mathcal{S}(\R)$, let $r(k)$ be the KdV reflection coefficient $r_{\KdV}(k)$ associated to $u_0(x) = u(x,0)$. In the generic case when $r(0) = i$, the solution $u(x,t)$ of the initial value problem for the KdV equation is given by
\begin{align*}
u(x,t) = 2i \partial_x \lim_{k\to \infty} \big(k \, M(x,t,k)\big)_{22}
\end{align*}
where $M$ is the unique solution of the following RH problem.  

\begin{RHproblem}[RH problem for KdV]\label{RHKdV}
Let $r: \mathbb{R}\to \mathbb{C}$ be a given function. 
Find a $2 \times 2$-matrix valued function $M(x,t,k)$ with the following properties:
\begin{enumerate}[$(a)$]
\item $M(x,t,\cdot) : \mathbb{C}\setminus \mathbb{R} \to \mathbb{C}^{2 \times 2}$ is analytic.

\item The limits of $M(x,t,k)$ as $k$ approaches $\mathbb{R}\setminus \{0\}$ from the left and right exist, are continuous on $\mathbb{R}\setminus \{0\}$, and satisfy
\begin{align}\label{jumpKdV}
& M_{+}(x,t,k) = M_{-}(x,t,k)v(x,t,k), \qquad k \in \mathbb{R}\setminus \{0\},
\end{align}
where $v$ is defined in terms of $r(k)$ by (\ref{vdefmKdV}).

\item As $k \to \infty$, $M(x,t,k) = I + k^{-1} M^{(1)}(x,t) + O(k^{-2})$, where the matrix $M^{(1)}$ satisfies $M_{12}^{(1)}  = 0$.

\item There exist complex-valued functions $\alpha(x,t), \beta(x,t), \gamma(x,t), \delta(x,t)$ such that
\begin{align}\label{MatzeroKdV}
M(x,t,k) = \frac{\alpha(x,t)}{k}\begin{pmatrix} 0 & 1 \\ 0 & i \end{pmatrix} + \begin{pmatrix}
\beta(x,t) & -i(\gamma(x,t) + \delta(x,t)) \\ 
i\beta(x,t) & \gamma(x,t) - \delta(x,t)
\end{pmatrix} + O(k) 
\end{align}
as $k \to 0$, $\im k \geq 0$.

\item $M$ satisfies the symmetries
\begin{align}\label{MsymmKdV}
M(x,t, k) = \mathcal{A} M(x,t,- k)\mathcal{A}^{-1} = \mathcal{B} \overline{M(x,t,\overline{k})}\mathcal{B}, \qquad k \in \C \setminus \mathbb{R},
\end{align}
where $\mathcal{A}$ and $\mathcal{B}$ are defined by 
\begin{align}\label{Acal B cal def KdV}
\mathcal{A} = \begin{pmatrix} 0 & -1 \\1 & 0 \end{pmatrix}, \qquad  \mathcal{B} = \begin{pmatrix} 0 & 1 \\1 & 0 \end{pmatrix}.
\end{align}
\end{enumerate}
\end{RHproblem}

\subsection{Properties of the reflection coefficients}\label{kdvrsubsec}
Suppose that $q_0$ and $u_0$ are solitonless initial data in the Schwartz class $\mathcal{S}(\R)$.
Then $r_\mKdV:\mathbb{R} \to \C$ and $r_\KdV:\mathbb{R} \to \C$ are both in the Schwartz class and obey the symmetry
\begin{align}
r(k) = -\overline{r(-k)} \quad \mbox{for } k \in \mathbb{R}.
\end{align}
For KdV, the simple pole of $\mathcal{U}$ at $k = 0$ implies that $s(k)$ has an expansion at the origin of the form
\begin{align}\label{s at 0 KdV}
s(k) = \frac{s^{(-1)}}{k} + s^{(0)} + s^{(1)}k + \cdots \qquad \mbox{as } k \to 0,
\end{align}
where, for some constant $\mathfrak{s}^{(-1)} \in \mathbb{R}$,
\begin{align}\label{spm2p KdV}
s^{(-1)} =  \mathfrak{s}^{(-1)} \begin{pmatrix}
-i & 1 \\
1 & i 
\end{pmatrix}.
\end{align}
For generic initial data, the coefficient $\mathfrak{s}^{(-1)}$ is non-zero; in this case, it is easy to see from \eqref{s at 0 KdV}--\eqref{spm2p KdV} that $r(0) =\lim_{k\to 0} -\frac{s_{21}(k)}{s_{22}(k)} = i$. The matrix $s$ satisfies $s(k) = \mathcal{A} s(- k)\mathcal{A}^{-1} = \mathcal{B} \overline{s(\overline{k})}\mathcal{B}$ and $\det s(k)=1$, from which we get $|r(k)|^{2}=1-1/|s_{22}(k)|^{2}$ for $k \in \mathbb{R}\setminus \{0\}$.
It follows that whereas the reflection coefficient for mKdV is everywhere smaller than $1$ (i.e., $|r_\mKdV(k)| < 1$ for all $k \in \mathbb{R}$), the reflection coefficient for KdV generically satisfies
\begin{align}\label{rkdvproperties}
 |r_\KdV(k)| < 1  \mbox{ for } k \in \mathbb{R}\setminus \{0\} \quad \mbox{and} \quad r_\KdV(0) = i.
\end{align}
This has the important implication that the Miura transformation $u = q_{x} + q^{2}$ maps Schwartz class solutions of mKdV onto a small (measure zero) subset of all Schwartz class solutions of KdV. More precisely, only solutions $u$ for which the associated coefficient $\mathfrak{s}^{(-1)}$ vanishes arise as images of mKdV Schwartz class solutions on the line. 

On the other hand, the RH problems \ref{RHmKdV} and \ref{RHKdV} can be formulated for any choice of $r(k)$ and generate solutions of mKdV and KdV, respectively, for a very large class of functions $r(k)$. 
Solutions corresponding to functions $r(k)$ which do not satisfy the above constraints, will not correspond to Schwartz class solutions, but to solutions with singularities and/or a lack of decay at spatial infinity. The Miura transformation is local in the sense that wherever a solution $q$ of mKdV is smooth, the Miura transformation $q \mapsto u = q_{x} + q^{2}$ is well-defined and generates a KdV solution $u$.

\subsection{The Miura correspondence}
Underlying the Miura transformation is a correspondence between RH problem \ref{RHmKdV} and \ref{RHKdV}. This correspondence exists for a large class of spectral functions $r(k)$. We will restrict ourselves to the class of functions $r \in \mathcal{S}(\R)$ such that $r(0) = i$ and $r(k) = -\overline{r(-k)}$ for $k \in \R$.
In view of (\ref{rkdvproperties}), this class contains the reflection coefficients relevant for generic Schwartz class solutions of KdV on the line. 
Our results are summarized in two theorems. The first theorem (Theorem \ref{mKdVtoKdVth}) describes how to go from mKdV to KdV, i.e., how to construct a solution $M$ of RH problem \ref{RHKdV} given a solution $m$ of RH problem \ref{RHmKdV}, and how the corresponding solutions of KdV to mKdV are related by the Miura transformation. The second theorem (Theorem \ref{KdVtomKdVth}) describes how to go in the opposite direction, from KdV to mKdV.
The analogy between Theorem \ref{mKdVtoKdVth} and Theorem \ref{mainth} should be evident.

\begin{theorem}[From mKdV to KdV]\label{mKdVtoKdVth}
Let $r \in \mathcal{S}(\R)$ be such that $r(k) = -\overline{r(-k)}$ for $k \in \R$ and $r(0) = i$.
Define the $2 \times 2$-matrix valued jump matrix $v(x,t,k)$ in terms of $r(k)$ by (\ref{vdefmKdV}). 
Let $\mathcal{N}$ be an open subset of $\mathbb{R}\times [0,+\infty)$ and suppose RH problem \ref{RHmKdV} has a (necessarily unique) solution $m(x,t,\cdot)$ for each $(x,t) \in \mathcal{N}$.

\begin{enumerate}[$(a)$]
\item \label{partamKdV}
The function $q$ defined by
\begin{align}\label{recoverqmKdV}
q(x,t) = 2 \, \lim_{k\to \infty} \big(k \, m(x,t,k)\big)_{12}
\end{align}
is smooth and real-valued on $\mathcal{N}$ and satisfies the mKdV equation (\ref{mKdV}) for $(x,t) \in \mathcal{N}$.

\item \label{partbmKdV}
Define $A(x,t)$ by
\begin{align}\label{AdefmKdV}
& A(x,t) = -\frac{q(x,t)}{2} \begin{pmatrix}
-i & 1 \\ 1 & i
\end{pmatrix}.
\end{align}
The $2\times 2$-matrix valued function $M(x,t,k)$ defined by
\begin{align}\label{MfrommdefmKdV}
M(x,t,k) = \bigg(I + \frac{A(x,t)}{k}\bigg)m(x,t,k)
\end{align}
satisfies RH problem \ref{RHKdV} for each $(x,t) \in \mathcal{N}$. 

\item \label{partcmKdV}
The function $u$ defined by
\begin{align}\label{recoveruKdV}
u(x,t) = 2i \partial_x \lim_{k\to \infty} \big(k \, M(x,t,k)\big)_{22}
\end{align}
is smooth and real-valued on $\mathcal{N}$ and satisfies the KdV equation (\ref{KdV}) for $(x,t) \in \mathcal{N}$.

\item \label{partdmKdV}
The solutions $u$ and $q$ are related by the Miura transformation
\begin{align}\label{miurakdv}
u(x,t) = q_x(x,t) + q(x,t)^2, \qquad (x,t) \in \mathcal{N}.
\end{align}

\end{enumerate}
\end{theorem}

\begin{theorem}[From KdV to mKdV]\label{KdVtomKdVth}
Let $r \in \mathcal{S}(\R)$ be such that $r(k) = -\overline{r(-k)}$ for $k \in \R$ and $r(0) = i$.
Define the $2 \times 2$-matrix valued jump matrix $v(x,t,k)$ in terms of $r(k)$ by (\ref{vdefmKdV}). 
Let $\mathcal{N}$ be an open subset of $\mathbb{R}\times [0,+\infty)$ and suppose RH problem \ref{RHKdV} has a (necessarily unique) solution $M(x,t,\cdot)$ for each $(x,t) \in \mathcal{N}$ and that the function $\delta(x,t)$ in (\ref{MatzeroKdV}) is nonzero on $\mathcal{N}$.

\begin{enumerate}[$(a)$]
\item \label{partaKdV}
Define $B(x,t)$ by
\begin{align}
& B(x,t) = \frac{\alpha(x,t)}{2\delta(x,t)} \begin{pmatrix}
-i & 1 \\ 1 & i
\end{pmatrix}, \qquad (x,t) \in \mathcal{N},
\end{align}
where $\alpha(x,t)$ and $\delta(x,t)$ are the functions in (\ref{MatzeroKdV}).
The $2\times 2$-matrix valued function $m(x,t,k)$ defined by
\begin{align}\label{mfromMdefKdV}
m(x,t,k) = \bigg(I + \frac{B(x,t)}{k}\bigg)M(x,t,k)
\end{align}
satisfies RH problem \ref{RHmKdV} for each $(x,t) \in \mathcal{N}$.

\item \label{partbKdV}
The function $u$ defined by (\ref{recoveruKdV}) is smooth and real-valued on $\mathcal{N}$ and satisfies the KdV equation (\ref{KdV}) for $(x,t) \in \mathcal{N}$.

\item \label{partcKdV}
The function $q$ defined by (\ref{recoverqmKdV}) is smooth and real-valued on $\mathcal{N}$ and satisfies the mKdV equation (\ref{mKdV}) for $(x,t) \in \mathcal{N}$. 

\item \label{partdKdV}
The solutions $u$ and $q$ are related by the Miura transformation
\begin{align}\label{miurakdv2}
u(x,t) = q_x(x,t) + q(x,t)^2, \qquad (x,t) \in \mathcal{N}.
\end{align}

\item \label{parteKdV}
For $(x,t) \in \mathcal{N}$, it holds that
\begin{align}\label{uqdeltaalpha}
\delta_{xx}(x,t) = u(x,t) \delta(x,t), \qquad q(x,t) = \frac{\delta_x(x,t)}{\delta(x,t)}, \qquad
\alpha(x,t) = \delta_x(x,t).
\end{align}

\end{enumerate}
\end{theorem}

\begin{remark}
Given $u$, the Miura transformation (\ref{miurakdv2}) can be viewed as a Ricatti equation for $q$. The relations in (\ref{uqdeltaalpha}) can be interpreted as the standard procedure for obtaining the solution of this Ricatti equation in terms of the solution of a second-order linear ordinary differential equation. Indeed, (\ref{uqdeltaalpha}) expresses the solution of $u = q_x + q^2$ as $q = \delta_x/\delta$ where $\delta$ is the solution of the second-order equation $\delta_{xx} = u \delta$.
\end{remark}

\begin{remark}
Let us comment on the assumption in Theorem  \ref{KdVtomKdVth} that the function $\delta(x,t)$ be nonzero on $\mathcal{N}$.
It is clear from (\ref{uqdeltaalpha}) that the inverse Miura transformation $u \mapsto q$ may introduce singularities at the points where $\delta$ vanishes. At these points RH problem \ref{RHmKdV} does not have a solution. For example, as discussed above, for $r(k)$ associated to generic Schwartz class solutions of KdV, there is a solution $M(x,t,k)$ of the RH problem for KdV for any $(x,t) \in \mathbb{R}\times [0,+\infty)$. But the corresponding solution of mKdV will not be in the Schwartz class (because if it were, we would have $|r(0)|<1$) and the best we can say is that it will be a smooth solution of mKdV on the set where $\delta(x,t)$ is nonzero.
\end{remark}



\subsection{Proof of Theorem \ref{mKdVtoKdVth}}
Suppose $r \in \mathcal{S}(\R)$ satisfies $r(k) = -\overline{r(-k)}$ for $k \in \R$ and $r(0) = i$, and assume that $m(x,t,\cdot)$ is a solution of RH problem \ref{RHmKdV} for all $(x,t)$ in some open subset $\mathcal{N}$ of $\mathbb{R}\times [0,+\infty)$.
Since $\det v = 1$, standard arguments show that the solution $m$ is unique. 
Moreover, since $r(k) = -\overline{r(-k)}$ for $k \in \mathbb{R}$, we have
\begin{align*}
v(x,t,k) =  \mathcal{A}v(x,t,- k)^{-1} \mathcal{A}^{-1} = \mathcal{B}\overline{v(x,t,\overline{k})}^{-1} \mathcal{B}, \qquad k \in \mathbb{R},
\end{align*}
so that $\mathcal{A} m(x,t,- k)\mathcal{A}^{-1}$ and $\mathcal{B} \overline{m(x,t,\overline{k})}\mathcal{B}$ also obey RH problem \ref{RHmKdV}; by uniqueness, 
\begin{align}\label{msymmmKdV}
m(x,t, k) = \mathcal{A} m(x,t,- k)\mathcal{A}^{-1} = \mathcal{B} \overline{m(x,t,\overline{k})}\mathcal{B}, \qquad k \in \C \setminus \mathbb{R},
\end{align}
By Cauchy's formula,
\begin{align}\label{mrepresentationmKdV}
m(x, t, k) = I + \frac{1}{2\pi i}\int_{\R} m_-(x, t, s) (v(x, t, s)-I)  \frac{ds}{s - k}.
\end{align}
It follows from (\ref{msymmmKdV}), (\ref{mrepresentationmKdV}), the fact that $r \in \mathcal{S}(\R)$, and standard theory for singular integral equations (see e.g. \cite[Chapter 2]{TO2016}) that
\begin{align}\label{matinftymKdV}
 m(x,t,k) = I + \frac{m^{(1)}(x,t)}{k} + \frac{m^{(2)}(x,t)}{k^2} + O(k^{-3}), \qquad k \to \infty,
\end{align}
where the coefficient matrices $\{m^{(j)}\}_1^2$ have the form
\begin{align}\label{m1m2coefficients KdV}
m^{(1)} = \begin{pmatrix} -m_{22}^{(1)} &  m_{12}^{(1)} \\
m_{12}^{(1)} & m_{22}^{(1)}
\end{pmatrix}, 
\qquad
m^{(2)} = \begin{pmatrix} m_{22}^{(2)} &  m_{12}^{(2)} \\
-m_{12}^{(2)} & m_{22}^{(2)}
\end{pmatrix}, \end{align}
where 
$m_{12}^{(1)}, m_{22}^{(2)}$ are real-valued and $m_{22}^{(1)}, m_{12}^{(2)}$ are purely imaginary.
In particular, $q = 2 m_{12}^{(1)}$ is well-defined by (\ref{recoverqmKdV}) and real-valued. Smoothness of $m$ (and hence also of $q$) as a function of $(x,t) \in \mathcal{N}$ follows from (\ref{mrepresentationmKdV}) and the rapid decay of $r(k)$.
Standard dressing arguments now show that $m(x,t,k)$ obeys the Lax pair equations
\begin{align}\label{laxmKdV}
  \begin{cases}
  m_x + ik[\sigma_3, m] = \mathcal{U} m, 
  	\\
  m_t + 4ik^3[\sigma_3, m] = \mathcal{V} m,
  \end{cases}
   \qquad (x,t) \in \mathcal{N}, \; k \in \C \setminus \R,
\end{align}
where $\mathcal{U}$ and $\mathcal{V}$ are defined in terms of $q$ by (\ref{UVdefmKdVKdV}).
In particular, $q$ satisfies (\ref{mKdV}) for $(x,t) \in \mathcal{N}$.
This completes the proof of part $(\ref{partamKdV})$.

To prove $(\ref{partbmKdV})$, define $A(x,t)$ and $M(x,t,k)$ as in (\ref{AdefmKdV})--(\ref{MfrommdefmKdV}).
It is immediate that $M(x,t,k)$ is analytic for $k \in \C \setminus \R$ and obeys the jump relation (\ref{jumpKdV}) on $\mathbb{R}\setminus \{0\}$.
By (\ref{matinftymKdV}) and (\ref{MfrommdefmKdV}), we have $M(x,t,k) = I + k^{-1} M^{(1)}(x,t) + O(k^{-2})$ as $k \to \infty$, where $M^{(1)} = m^{(1)} + A$ satisfies $M^{(1)}_{12} = m^{(1)}_{12} + A_{12} = 0$.
The symmetry relations (\ref{msymmmKdV}) and $A(x,t) = - \mathcal{A} \, A(x,t)\mathcal{A}^{-1}
= \mathcal{B}\overline{A(x,t)}\mathcal{B}$ imply that $M$ obeys (\ref{MsymmKdV}).

We next show that there exist functions $\alpha(x,t)$, $\beta(x,t)$, $\gamma(x,t)$, $\delta(x,t)$ such that (\ref{MatzeroKdV}) holds.
We deduce from (\ref{mrepresentationmKdV}), the fact that $r \in \mathcal{S}(\R)$, and standard theory for singular integral equations that there exist matrices $\mathfrak{m}^{(0)}(x,t)$ and $\mathfrak{m}^{(1)}(x,t)$ such that
\begin{align*}
m(x,t,k) = \mathfrak{m}^{(0)}(x,t) + \mathfrak{m}^{(1)}(x,t)k + O(k^2) \qquad \text{as}\ k \to 0, \ \im k \geq 0.
\end{align*}
Substituting this expansion into (\ref{MfrommdefmKdV}), we obtain
\begin{align}\label{MAm0}
M(x,t,k) =  \frac{A \mathfrak{m}^{(0)} }{k} + (\mathfrak{m}^{(0)} + A\mathfrak{m}^{(1)}) + O(k), \qquad k \to 0, \ \im k \geq 0.
\end{align}
By (\ref{msymmmKdV}), the leading coefficient obeys $\mathfrak{m}^{(0)} = \mathcal{A}\mathfrak{m}^{(0)}\mathcal{A}^{-1}v(0)$.
Since $r(0) = i$, we have
\begin{align}\label{vat0mKdV}
& v(x,t,0) = \begin{pmatrix} 0 & i \\ i & 1 \end{pmatrix},  \quad \text{and hence} \quad \mathfrak{m}^{(0)} = \begin{pmatrix} \mathfrak{m}^{(0)}_{11} & i(\mathfrak{m}^{(0)}_{22}-\mathfrak{m}^{(0)}_{11}) \\ 
i \mathfrak{m}^{(0)}_{11} & \mathfrak{m}^{(0)}_{22}
\end{pmatrix}.
\end{align}
Substitution into (\ref{MAm0}) shows that (\ref{MatzeroKdV}) holds with $\alpha := \frac{1}{2}(\mathfrak{m}^{(0)}_{11} - 2\mathfrak{m}^{(0)}_{22})q$, $\beta := \mathfrak{m}^{(0)}_{11} + \frac{iq}{2}(\mathfrak{m}^{(1)}_{11} + i \mathfrak{m}^{(1)}_{21})$, and appropriate choices of $\gamma$ and $\delta$.
This completes the proof of $(\ref{partbmKdV})$.

Let us finally prove $(\ref{partcmKdV})$ and $(\ref{partdmKdV})$. 
From (\ref{laxmKdV}), we see that $2i\partial_x m^{(1)}_{22} = q^2$. Hence, defining $u$ by (\ref{recoveruKdV}), we find $u = 2i\partial_x M^{(1)}_{22} = 2i\partial_x( m^{(1)}_{22} + A_{22})
 = q^2 + q_x$. This proves $(\ref{partdmKdV})$ and, since $q$ satisfies (\ref{mKdV}), it also implies that $u$ satisfies (\ref{KdV}).
This completes the proof of Theorem \ref{mKdVtoKdVth}.

\subsection{Proof of Theorem \ref{KdVtomKdVth}}
Suppose $r \in \mathcal{S}(\R)$ satisfies $r(k) = -\overline{r(-k)}$ for $k \in \R$ and $r(0) = i$, and assume that $M(x,t,\cdot)$ is a solution of RH problem \ref{RHKdV} for all $(x,t)$ in some open subset $\mathcal{N}$ of $\mathbb{R}\times [0,+\infty)$.

We first show that the solution $M$ is unique. The symmetries (\ref{MsymmKdV}) and the fact that $M_{12}^{(1)}  = 0$ imply that
\begin{align}\label{MatinftyKdV} 
M(x,t,k) = & \; I + \frac{M_{22}^{(1)}}{k} \begin{pmatrix} -1 & 0 \\ 
0 & 1
\end{pmatrix}  + O(k^{-2}), \qquad k \to \infty, \; (x,t) \in \mathcal{N},
\end{align}
where $iM_{22}^{(1)}$ is a real-valued function of $(x,t)\in \mathcal{N}$.
It follows from (\ref{MatinftyKdV}) and (\ref{MatzeroKdV}) that
\begin{align*}
& \det M(x,t,k) = 1+O(k^{-2}) & & \mbox{as } k \to \infty, \\
& \det M(x,t,k) = O(1) & & \mbox{as } k \to 0.
\end{align*}
Since $\det M$ is analytic for $k \in \C \setminus \{0\}$, we infer that $M$ has unit determinant. 
Suppose that $M$ and $N$ are two solutions of RH problem \ref{RHKdV}. By the above argument, $\det M$ and $\det N$ are identically equal to one. In particular, the inverse $N^{-1}$ exists and is given by
\begin{align*}
N^{-1} = \begin{pmatrix}
N_{22} & -N_{12} \\
-N_{21} & N_{11}
\end{pmatrix}.
\end{align*}
Expanding this expression for $N^{-1}$ as $k \to 0$, $\im k \geq 0$, and using (\ref{MatzeroKdV}), we find 
\begin{align}\label{N1Aat0 KdV}
N^{-1}(x,t,k) = 
\frac{\alpha_N(x,t)}{k} 
\begin{pmatrix}
i & -1 \\
0 & 0
\end{pmatrix} 
+ \beta_N(x,t)
\begin{pmatrix}
0 & 0 \\
-i & 1
\end{pmatrix} 
+ \begin{pmatrix}
O(1) & O(1) \\
O(k) & O(k)
\end{pmatrix}
\end{align}
as $k \to 0$, $\im k \geq 0$. Similarly, expanding the expression for $N^{-1}$ as $k \to \infty$ and using (\ref{MatinftyKdV}), we find 
\begin{align}\nonumber
N^{-1}(x,t,k) = &\; I + 
\frac{N_{22}^{(1)}(x,t)}{k} 
\begin{pmatrix} 
1 & 0 \\ 
0 & -1
\end{pmatrix}
+ O(k^{-2}), \qquad k\to \infty.
\end{align}
By (\ref{MatzeroKdV}) and (\ref{N1Aat0 KdV}), we have, as $k$ approaches $0$, $\im k \geq 0$,
\begin{align}\label{lol5}
MN^{-1} = \frac{\alpha_N \beta_M - \alpha_M \beta_N}{k}\begin{pmatrix}
i & -1 \\ -1 & -i
\end{pmatrix} + O(1),
\end{align}
showing that $MN^{-1}$ has at most a simple pole at $k = 0$. On the other hand, expanding $MN^{-1}$ as $k \to \infty$ we get
\begin{align}\label{lol6}
M(x,t,k)N(x,t,k)^{-1} = I + \frac{M_{22}^{(1)}-N_{22}^{(1)}}{k} \begin{pmatrix}
-1 & 0 \\ 0 & 1
\end{pmatrix} + O(k^{-2}).
\end{align}
Since $MN^{-1}$ is entire, Liouville's theorem together with \eqref{lol5} and \eqref{lol6} show that 
$$M(x,t,k)N(x,t,k)^{-1} = I + \frac{M_{22}^{(1)}-N_{22}^{(1)}}{k} \begin{pmatrix}
-1 & 0 \\ 0 & 1
\end{pmatrix},$$
but this is only compatible with (\ref{lol5}) if $M_{22}^{(1)}-N_{22}^{(1)} = 0$, i.e., if $MN^{-1} = I$. This proves that the solution $M$ is unique.

We next show that the functions $\alpha, \beta$, and $\delta$ in (\ref{MatzeroKdV}) are real-valued.
Using the symmetry \eqref{MsymmKdV} and suppressing the $(x,t)$-dependence for conciseness, we obtain
\begin{align}\label{MBMvBKdV}
M_+(k) & = \mathcal{B}\overline{M_+(\overline{k})}\overline{v(\overline{k})^{-1}}\mathcal{B}
= \mathcal{B}\overline{M_+(\overline{k})}\mathcal{B}v(k), \qquad k \in \R \setminus \{0\}.
\end{align}
Inserting expansion (\ref{MatzeroKdV}), the terms of $O(k^{-1})$ and $O(1)$ of (\ref{MBMvBKdV}) yield the relations
\begin{align*}
& \alpha\begin{pmatrix} 0 & 1 \\ 0 & i \end{pmatrix}
= \mathcal{B} \bar{\alpha} \begin{pmatrix} 0 & 1 \\ 0 & -i \end{pmatrix}\mathcal{B}v(0),
	\\
& \begin{pmatrix}
\beta & -i( \gamma + \delta ) \\ 
i\beta & \gamma - \delta
\end{pmatrix}
= \mathcal{B} \bar{\alpha} \begin{pmatrix} 0 & 1 \\ 0 & -i \end{pmatrix}\mathcal{B}v'(0)
+ \mathcal{B} \begin{pmatrix}
\bar{\beta} & i(\bar{\gamma} + \bar{\delta}) \\ 
-i\bar{\beta} & \bar{\gamma} - \bar{\delta}
\end{pmatrix} \mathcal{B} v(0).
\end{align*}
Since 
\begin{align*}
 v(0) =&\; \begin{pmatrix} 0 & i \\ i & 1 \end{pmatrix},
\qquad v'(0) = \begin{pmatrix} 0 & 2x - r'(0) \\ - 2x + r'(0) & 0 \end{pmatrix}.
\end{align*}
these relation imply that $\alpha$, $\beta$, $\delta$ are real-valued. 

Assume now that the function $\delta(x,t)$ in (\ref{MatzeroKdV}) is nonzero on $\mathcal{N}$.
Then $m(x,t,k)$ is well-defined by (\ref{mfromMdefKdV}) for $(x,t) \in \mathcal{N}$.
Clearly, $m(x,t,k)$ is analytic for $k \in \C \setminus \R$, obeys the jump relation (\ref{jumpmKdV}) on $\mathbb{R}\setminus \{0\}$, and satisfies $m(x,t,k) = I + O(k^{-1})$ as $k \to \infty$. As a consequence of the reality of $\alpha$ and $\delta$, we have $B(x,t) = - \mathcal{A} B(x,t)\mathcal{A}^{-1}
= \mathcal{B}\overline{B(x,t)}\mathcal{B}$. It follows that $m$ obeys the symmetries in (\ref{msymmmKdV}). Moreover, as $k \to 0$, $\im k \geq 0$, we have
\begin{align*}
m(x,t,k) = &\; \bigg(\frac{B(x,t)}{k} + I\bigg)\bigg\{ \frac{\alpha(x,t)}{k}\begin{pmatrix} 0 & 1 \\ 0 & i \end{pmatrix} + \begin{pmatrix}
\beta(x,t) & -i(\gamma(x,t) + \delta(x,t)) \\ 
i\beta(x,t) & \gamma(x,t) - \delta(x,t)
\end{pmatrix} + O(k)\bigg\}
	\\
= &\; \frac{1}{k}\bigg\{\alpha(x,t)\begin{pmatrix} 0 & 1 \\ 0 & i \end{pmatrix}
+ B(x,t) \delta(x,t) \begin{pmatrix}
0 & -i \\
0 & -1
\end{pmatrix} \bigg\} + O(1)
= O(1).
\end{align*}
Thanks to (\ref{msymmmKdV}), $m$ is non-singular also as $k \to 0$ from the lower half-plane.
It follows that $m(x,t,\cdot)$ satisfies RH problem \ref{RHmKdV} for every $(x,t) \in \mathcal{N}$, completing the proof of $(\ref{partaKdV})$.
An application of Theorem \ref{mKdVtoKdVth} then immediately yields also assertions $(\ref{partbKdV})$--$(\ref{partdKdV})$.
Comparing (\ref{MfrommdefmKdV}) and (\ref{mfromMdefKdV}), we conclude that
$$I + \frac{A(x,t)}{k} = \bigg(I + \frac{B(x,t)}{k}\bigg)^{-1},$$
which holds if and only if $q = \alpha/\delta$. On the other hand, substituting $m = (I + B/k)M$ and the expansion (\ref{MatzeroKdV}) of $M$ into the $x$-part of the Lax pair (\ref{laxmKdV}) and considering the terms of $O(1)$ of $m_{12x} + i m_{22x}$, we obtain the relation $\delta_x = q\delta$. We conclude that $\delta_x = \alpha$ and $q = \delta_x/\delta$; substitution of $q = \delta_x/\delta$ into the Miura transformation (\ref{miurakdv}) then gives $u = \delta_{xx}/\delta$.
This completes the proof of Theorem \ref{KdVtomKdVth}.

\appendix

\section{The initial value problems}\label{initialapp}
We recall the origin of RH problem \ref{RHgoodboussinesq} and \ref{RH3x3} and how they can be used to solve the initial value problems for (\ref{boussinesqsystem}) and (\ref{3x3eq}) on the real line. We also discuss the origin of Assumptions \ref{r1r2assumption} and \ref{r1r2at0assumption} on the spectral functions $r_1$ and $r_2$. The material in this appendix can be found in \cite{L3x3, CL2021, CL2022}.

\subsection{The initial value problem for (\ref{3x3eq})} 
Let $\omega = e^{\frac{2\pi i}{3}}$, define $\{l_j(k), z_j(k)\}_{j=1}^3$ by \eqref{lmexpressions}, and define $J$, $\mathcal{L}$, and $\mathcal{Z}$ by
\begin{align*}
J = \mbox{diag}(\omega, \omega^{2}, 1), \qquad \mathcal{L} = k J = \mbox{diag}(l_{1},l_{2},l_{3}), \qquad \mathcal{Z} = k^{2} J^{2} = \mbox{diag}(z_{1},z_{2},z_{3}).
\end{align*}
Equation (\ref{3x3eq}) is the compatibility condition of the Lax pair equations \cite{L3x3}
\begin{align}\label{3x3lax}
\begin{cases}
X_x-[\mathcal{L},X]=\mathcal{U}X,\\
X_t-[\mathcal{Z},X]=\mathcal{V}X,
\end{cases}
\end{align}
where $k \in \C$ is the spectral parameter and
\begin{align}
\mathcal{U} = &\; \begin{pmatrix} 0 & (1-\omega^2) \bar{q} & (1-\omega)q \\
(1-\omega)q & 0 & (1-\omega^2)\bar{q} \\
(1-\omega^2)\bar{q} & (1-\omega)q & 0 \end{pmatrix}, \label{expression for Ufrak}
	\\ \nonumber
 \mathcal{V} = &\; k \begin{pmatrix} 
0 & (\omega^2-1) \bar{q} & (1-\omega^2)q \\
(\omega-1)q & 0 & (1-\omega)\bar{q} \\
(\omega-\omega^2)\bar{q} & (\omega^2-\omega)q & 0 \end{pmatrix}
	\\
& + (\bar{q}_x - 3q^2)\begin{pmatrix} 0 & \omega & 0 \\ 0 & 0 & \omega \\ \omega & 0 & 0 \end{pmatrix}
 + (q_x - 3\bar{q}^2)\begin{pmatrix} 0 & 0 & \omega^2 \\ \omega^2 & 0 & 0 \\ 0 & \omega^2 & 0 \end{pmatrix}. \label{expression for Vfrak}
\end{align}
Given initial data $q_0(x) = q(x,0)$ in $\mathcal{S}(\R)$,  we define the spectral functions $r_1(k)$  and $r_2(k)$ as follows. Let $X(x,k)$ and $X^A(x,k)$ be the unique solutions of the Volterra integral equations
\begin{subequations}\label{XXAdef}
\begin{align}  
 & X(x,k) = I - \int_x^{\infty} e^{(x-x')\widehat{\mathcal{L}(k)}} (\mathcal{U}_0 X)(x',k) dx',
	\\\label{XXAdefb intro}
 & X^A(x,k) = I + \int_x^{\infty} e^{-(x-x')\widehat{\mathcal{L}(k)}} (\mathcal{U}_0^T X^A)(x',k) dx',	
\end{align}
\end{subequations}
where $\mathcal{U}_0(x,k) = \mathcal{U}(x,0,k)$ and $\widehat{\mathcal{L}}$ denotes the operator which acts on a $3 \times 3$ matrix $A$ by $\widehat{\mathcal{L}}A = [\mathcal{L}, A]$ (i.e., $e^{\widehat{\mathcal{L}}}A = e^\mathcal{L} A e^{-\mathcal{L}}$), and $\mathcal{U}_0^T$ denotes the transpose of $\mathcal{U}_0$. Let
\begin{align}\label{sdef}
& s(k) := I - \int_\R e^{-x\widehat{\mathcal{L}(k)}}(\mathcal{U}_{0}X)(x,k)dx,
 	\qquad s^A(k) := I + \int_\R e^{x\widehat{\mathcal{L}(k)}}(\mathcal{U}_{0}^T X^A)(x,k)dx.
\end{align}
The two spectral functions $\{r_j(k)\}_1^2$ are defined by
\begin{align}\label{3x3r1r2def}
\begin{cases}
r_1(k) = \frac{(s(k))_{12}}{(s(k))_{11}}, & k \in (0,\infty),
	\\ 
r_2(k) = \frac{(s^A(k))_{12}}{(s^A(k))_{11}}, \quad & k \in (-\infty,0).
\end{cases}
\end{align}	
The functions $(s(k))_{11}$ and $(s^A(k))_{11}$ which appear in the denominators of \eqref{3x3r1r2def} are analytic in $D_{1}$ and $D_{4}$, respectively, see \cite{CL2021}. 
We refer to the initial data $q_0$ as {\it solitonless} if $(s(k))_{11}$ and $(s^A(k))_{11}$ are nonzero for $k \in \bar{D}_1$ and $k \in \bar{D}_4$, respectively.

It is shown in \cite[Lemma 2.2]{CL2021} that the 
spectral functions $\{r_j(k)\}_1^2$ satisfy the properties in Assumption \ref{r1r2assumption}
for any solitonless initial data in the Schwartz class.
Moreover, it is shown in \cite[Theorem 1]{CL2021} that if $q(x,t)$ is a Schwartz class solution of (\ref{3x3eq}) on $\R \times [0,T)$ with solitonless initial data $q_0(x) = q(x,0)$, then RH problem \ref{RH3x3} has a unique solution $m(x,t,k)$ for any $(x,t) \in \R \times [0,T)$ and the solution $q(x,t)$ can be obtained for any $(x,t) \in \R \times [0,T)$ from the formula
$$q(x,t) = \lim_{k\to\infty} k \, m_{13}(x,t,k).$$

\begin{remark}
  The spectral functions $\{r_j(k)\}_1^2$ associated to solitonless Schwartz class initial data $q_0$ for (\ref{3x3eq}) do {\upshape not} satisfy the properties in Assumption \ref{r1r2at0assumption}. Instead they satisfy $|r_1(0)| < 1$ and $|r_2(0)| < 1$, see \cite[Lemma 2.2]{CL2021}. An analogous situation occurs also for mKdV and KdV, where the reflection coefficient $r(k)$ for mKdV satisfies $|r(0)| < 1$, whereas the reflection coefficient for KdV generically satisfies $r(0) = i$, see Section \ref{kdvrsubsec}.
\end{remark}

\subsection{The initial value problem for (\ref{boussinesqsystem})} 
The system (\ref{boussinesqsystem}) also admits a Lax pair of the form (\ref{3x3lax}). More precisely, (\ref{boussinesqsystem}) is the compatibility condition of (\ref{3x3lax}) provided that $\mathcal{U}$ and $\mathcal{V}$ are given by (see \cite{CL2022})
\begin{subequations}\label{UVexpressions}
\begin{align}
& \mathcal{U}(x,t,k) = \frac{\mathcal{U}^{(2)}(x,t)}{k^{2}} + \frac{\mathcal{U}^{(1)}(x,t)}{k}, \label{expression for U} \\
& \mathcal{V}(x,t,k) = \frac{\mathcal{V}^{(2)}(x,t)}{k^{2}} + \frac{\mathcal{V}^{(1)}(x,t)}{k} + \mathcal{V}^{(0)}(x,t), \label{expression for V}
\end{align}
\end{subequations}
with
\begin{align*}
& \mathcal{U}^{(2)} = - \frac{v+u_{x}}{3} \begin{pmatrix}
\omega & \omega^{2} & 1 \\
\omega & \omega^{2} & 1 \\
\omega & \omega^{2} & 1 
\end{pmatrix}, \qquad \mathcal{U}^{(1)} = - \frac{2u}{3}\begin{pmatrix}
\omega^{2} & \omega & 1 \\
\omega^{2} & \omega & 1 \\
\omega^{2} & \omega & 1 
\end{pmatrix}, \\
& \mathcal{V}^{(2)} = \frac{-3v_{x}+u_{xx}}{9}\begin{pmatrix}
\omega & \omega^{2} & 1 \\
\omega & \omega^{2} & 1 \\
\omega & \omega^{2} & 1
\end{pmatrix}, \qquad \mathcal{V}^{(0)} = \frac{2u}{3} \begin{pmatrix}
0 & \omega & \omega^{2} \\
\omega^{2} & 0 & \omega \\
\omega & \omega^{2} & 0
\end{pmatrix}, \\
& \mathcal{V}^{(1)} = \frac{v}{3}\begin{pmatrix}
-2\omega^{2} & \omega^{2} & \omega^{2} \\
\omega & -2 \omega & \omega \\
1 & 1 & -2 
\end{pmatrix} + \frac{(1-\omega)u_{x}}{9}\begin{pmatrix}
0 & 1 & -1 \\
-\omega^{2} & 0 & \omega^{2} \\
\omega & -\omega & 0
\end{pmatrix}.
\end{align*}
The functions $r_1(k)$ and $r_2(k)$ are still defined by (\ref{3x3r1r2def}) but with $\mathcal{U}_0(x,k) = \mathcal{U}(x,0,k)$ given by (\ref{UVexpressions}). We refer to the initial data $\{u_0(x) = u(x,0), v_0(x) = v(x,0)\}$ as {\it solitonless} if $(s(k))_{11}$ and $(s^A(k))_{11}$ are nonzero for $k \in \bar{D}_1 \setminus \{0\}$ and $k \in \bar{D}_4 \setminus \{0\}$, respectively.

The double poles of $\mathcal{U}$ and $\mathcal{V}$ in (\ref{UVexpressions}) imply that each of the four functions $s_{11}$, $s_{12}$, $s^A_{11}$, and $s^A_{12}$ that appear in (\ref{3x3r1r2def}) has at most a double pole at $k = 0$ \cite{CL2022}. Furthermore, $s_{12}$ has a double pole if and only if $s_{11}$ has a double pole, and $s^A_{12}$ has a double pole if and only if $s^A_{11}$ has a double pole. We say that the initial data $\{u_0(x), v_0(x)\}$ is {\it generic} if the behavior of these functions is generic at $k = 0$ in the sense that
$$\lim_{k \to 0} k^2 (s(k))_{11} \neq 0, \qquad \lim_{k \to 0} k^2 (s^A(k))_{11} \neq 0.$$

It is shown in \cite[Theorem 2.3]{CL2022} that, for any generic solitonless initial data in the Schwartz class, the spectral functions $\{r_j(k)\}_1^2$ satisfy the properties in Assumptions \ref{r1r2assumption} and \ref{r1r2at0assumption}.
Moreover, it is shown in \cite[Theorem 2.6]{CL2022} that if $\{u(x,t), v(x,t)\}$ is a Schwartz class solution of (\ref{boussinesqsystem}) on $\R \times [0,T)$ with generic solitonless initial data, then RH problem \ref{RHgoodboussinesq} has a unique solution $M(x,t,k)$ for any $(x,t) \in \R \times [0,T)$ and the solution $\{u(x,t), v(x,t)\}$ can be obtained for any $(x,t) \in \R \times [0,T)$ from the formulas
\begin{align}\label{recoveruv}
\begin{cases}
 \displaystyle{u(x,t) = -\frac{3}{2}\frac{\partial}{\partial x}\lim_{k\to \infty}k\big(M_{33}(x,t,k) - 1\big),}
	\vspace{.1cm}\\
 \displaystyle{v(x,t) = -\frac{3}{2}\frac{\partial}{\partial t}\lim_{k\to \infty}k\big(M_{33}(x,t,k) - 1\big).}
\end{cases}
\end{align}

\section{Soliton solutions of (\ref{3x3eq})}\label{solitonapp}
In this appendix, we construct soliton solutions of (\ref{3x3eq}); in particular, we derive the exact solutions in (\ref{qonesoliton}) and (\ref{qbreather}).

Soliton solutions of (\ref{3x3eq}) can be generated by adding poles to RH problem \ref{RH3x3} and letting the jump matrix $v$ be the identity matrix. 
To see how to add the poles, let $X(x,t,k)$ be the solution of the Volterra integral equation
\begin{align} \label{Xdef}
 & X(x,t,k) = I - \int_x^{\infty} e^{(x-x')kJ} (\mathcal{U} X)(x',t,k) e^{-(x-x')kJ} dx'.
\end{align}
For compactly supported initial data, the solution $m(x,t,k)$ of RH problem \ref{RH3x3} is related to $X$ for $k \in D_1$ via 
\begin{align}\label{mXT1}
m(x,t,k) = X(x,t,k) e^{xkJ+tk^{2}J^2} T_1(k) e^{-xkJ-tk^{2}J^{2}}, 
\end{align}
where $T_1(k)$ is given in terms of the entries of the scattering matrix $s(k)$ in (\ref{sdef}) by
\begin{align}\nonumber
 T_1(k) = \begin{pmatrix}
  1 & -\frac{s_{12}}{s_{11}} & \frac{m_{31}(s)}{m_{33}(s)} \\
 0 & 1 & -\frac{m_{32}(s)}{m_{33}(s)} \\
 0 & 0 & 1
 \end{pmatrix}
\end{align}
with $m_{ij}(s)$ the $(ij)$th minor of the matrix $s$, see \cite{CL2021} (see also \cite[Lemma 4.4]{CL2022} for a similar situation with more details provided).
The relation (\ref{mXT1}) implies that the three columns $[m]_j$, $j = 1,2,3$, of $m(x,t,k)$ obey the following relations in $D_1$:
\begin{subequations}
\begin{align}\label{MXYcolumnsd}
& [m]_1 = [X]_1,
	\\\label{MXYcolumnse}
& [m]_2 = -\frac{s_{12}}{s_{11}} e^{\theta_{12}} [X]_1 + [X]_2,
	\\ \label{MXYcolumnsf}
& [m]_3 = \frac{m_{31}(s)}{m_{33}(s)}e^{\theta_{13}} [X]_1
-\frac{m_{32}(s)}{m_{33}(s)} e^{\theta_{23}} [X]_2 + [X]_3.	
\end{align}
\end{subequations}

\subsection{One-solitons}
One-soliton solutions arise when $s_{11}(k)$ has a simple zero at some point $k_0 > 0$. 

Suppose $s_{11}$ has a simple zero at $k_0 > 0$ and that $s_{12}(k_0) \neq 0$. 
According to (\ref{MXYcolumnse}), it follows that $[m(x,t,k)]_2$ has a simple pole at $k_0$ satisfying the residue condition
$$\underset{k = k_0}{\res}[m(x,t,k)]_2 = e_1 e^{\theta_{12}(x,t,k_0)} [m(x,t,k_0)]_1,$$
where $e_1 := -\frac{s_{12}(k_0)}{\dot{s}_{11}(k_0)} \in \C$ is a constant.
Using the symmetry $m(x,t,k) = \mathcal{A}m(x,t,\omega k) \mathcal{A}^{-1}$, we conclude that $[m]_1$ and $[m]_3$ have simple poles at $\omega k_0$ and $\omega^2 k_0$, respectively, and that the following residue conditions hold:
\begin{subequations}
\begin{align}
& \underset{k = \omega k_0}{\res} [m(x,t,k)]_1 = \omega e_1 e^{\theta_{12}(x,t,k_0)} [m(x,t, \omega k_0)]_3,
	\\ 
& \underset{k = \omega^2 k_0}{\res} [m(x,t, k)]_3 = \omega^2 e_1 e^{\theta_{12}(x,t,k_0)} [m(x,t,\omega^2 k_0)]_2.
\end{align}
\end{subequations}
Together with the normalization condition $m(x,t,k) = I + O(k^{-1})$ as $k \to \infty$ this gives
\begin{align}\nonumber
& [m(x,t,k)]_1 = \begin{pmatrix} 1 \\ 0 \\ 0 \end{pmatrix}
 + \frac{ \omega e_1 e^{\theta_{12}(x,t,k_0)} [m(x,t,\omega k_0)]_3}{k - \omega k_0},
	\\\nonumber
& [m(x,t,k)]_2 =  \begin{pmatrix} 0 \\ 1 \\ 0 \end{pmatrix} + \frac{e_1 e^{\theta_{12}(x,t,k_0)} [m(x,t,k_0)]_1}{k -  k_0},
	\\\label{m1m2m3}
& [m(x,t,k)]_3 =  \begin{pmatrix} 0 \\ 0 \\ 1 \end{pmatrix} + \frac{\omega^2 e_1 e^{\theta_{12}(x,t,k_0)} [m(x,t,\omega^2 k_0)]_2}{k - \omega^2 k_0}.
\end{align}
We evaluate the first equation at $k = k_0$, the second equation at $k = \omega^2 k_0$, and the third equation at $k = \omega k_0$. 
The top entries of the equations in (\ref{m1m2m3}) then give
\begin{align*}
 m_{11}(x,t,k_0) 
 & = 1
 + \frac{ \omega e_1 e^{\theta_{12}(x,t,k_0)} m_{13}(x,t,\omega k_0)}{k_0 - \omega k_0},
	\\
m_{12}(x,t,\omega^2 k_0) & =  \frac{e_1 e^{\theta_{12}(x,t,k_0)} m_{11}(x,t,k_0)}{\omega^2 k_0 -  k_0},
	\\
m_{13}(x,t,\omega k_0) & = \frac{\omega^2 e_1 e^{\theta_{12}(x,t,k_0)} m_{12}(x,t,\omega^2 k_0)}{\omega k_0 -  \omega^2 k_0}.
\end{align*}
This gives three algebraic equations which we can solve for the three unknowns $m_{11}(x,t,k_0)$, $m_{12}(x,t,\omega^2 k_0)$, and $m_{13}(x,t,\omega k_0)$. Similarly, the middle entries of the equations in (\ref{m1m2m3}) can be solved for $m_{21}(x,t,k_0)$, $m_{22}(x,t,\omega^2  k_0)$, and $m_{23}(x,t,\omega k_0)$. 
This gives explicit expressions for
$$q(x,t) := \lim_{k\to \infty} k \, m_{13}(x,t,k)  = \omega^2 e_1 e^{\theta_{12}(x,t,k_0)} m_{12}(x,t,\omega^2 k_0)$$
and
$$\bar{q}(x,t) := \lim_{k\to \infty} k \, m_{23}(x,t,k)  = \omega^2 e_1 e^{\theta_{12}(x,t,k_0)} m_{22}(x,t,\omega^2 k_0),$$
 which can be verified to satisfy (\ref{3x3eq}) for any choice of $e_1 \in \C$ and $k_0 \in \R$. (The above definitions of $q$ and $\bar{q}$ follow from \cite[Eq.\ (3.11a)]{CL2021}.)
However, in the above computation, we have not taken the $\mathcal{B}$-symmetry of (\ref{regRHsymm}) into account; thus $\bar{q}$  is in general not the complex conjugate of $q$. 
A direct calculation shows that $\bar{q}$ is the complex conjugate of $q$ if and only if
$$|e_1|^2 = 3k_0^2.$$
Writing $e_1 = \sqrt{3} k_0 e^{i \arg e_1}$, this leads to the following class of (singular) one-soliton solutions of (\ref{3x3eq}):
\begin{align}\label{q3x3soliton}
q(x,t) = -\frac{1}{2} \frac{\sqrt{3} k_0 \omega^2 e^{\frac{i}{2} \left(\sqrt{3}
   k_0 (x-k_0 t)+ \arg(e_1) +\frac{\pi }{6}\right)} }{\sin\left(\frac{3}{2}
   \left(\sqrt{3} k_0 (x-k_0 t)+ \arg(e_1) + \frac{\pi }{6}\right)\right)}.
\end{align} 
Letting $c := k_0$ and $\arg(e_1) = -\pi/6$, we obtain the solution in (\ref{qonesoliton}). 

\begin{remark}
The one-soliton solutions we found in (\ref{q3x3soliton}) are singular. This is expected because for regular solutions of (\ref{3x3eq}) which decay as $x \to \pm \infty$, $s_{11}(k)$ cannot vanish for $k \geq 0$. Indeed, for such solutions, $|s_{12}(k)/s_{11}(k)| \leq 1$ for $k \geq 0$ by \cite[Eq. (A.2)]{CL2021}; hence, if $s_{11}(k_0) = 0$ for some $k_0 \geq 0$, we must have $s_{12}(k_0) = 0$. But if $s_{11}(k_0) = s_{12}(k_0) = 0$, then the $\mathcal{B}$-symmetry $s(k) = \mathcal{B} \overline{s(\overline{k})}\mathcal{B}$ implies that $s_{22}$ and $s_{21}$ also vanish at $k_0$. Consequently, $\det s(k_0) = 0$, which contradicts the fact that $s(k)$ has unit determinant. This means that (\ref{3x3eq}) does not admit regular one-solitons with decay at spatial infinity. 
\end{remark}

\subsection{Breather solitons}
Breather solitons arise when $s_{11}(k)$ has a simple nonreal zero. 

Suppose $k_0$ is a simple zero of $s_{11}$ in the open set $D_1$ so that $\im k_0 > 0$.
Let us consider the case where $s_{21}(k)$ also has a simple zero at $k_0$. Then $m_{33}(s) = s_{11}s_{22} - s_{12} s_{21}$ and $m_{32}(s) = s_{11}s_{23} - s_{13}s_{21}$ both vanish at $k_0$.
Taking the residues of (\ref{MXYcolumnse}) and (\ref{MXYcolumnsf}) at $k_0$ and suppressing the $(x,t)$-dependence for brevity, we obtain
\begin{align*}
& \underset{k = k_0}{\res}[m]_2 = C_1 [m(k_0)]_1,
	\qquad
\underset{k = k_0}{\res} [m]_3 =  C_2 [m(k_0)]_1,	
\end{align*}
where $C_1 = C_1(x,t)$ and $C_2 = C_2(x,t)$ are given by
$$ C_1 := -\underset{k = k_0}{\res}\frac{s_{12}(k)}{s_{11}(k)} e^{\theta_{12}(x,t,k_0)}, \qquad
C_2 := \underset{k = k_0}{\res} \frac{m_{31}(s(k))}{m_{33}(s(k))} e^{\theta_{13}(x,t,k_0)}.$$
We write these conditions as
\begin{subequations}\label{mCDexpansions}
\begin{align}
m(k) \begin{pmatrix} 0 \\ 1 \\ 0 \end{pmatrix} = C_1 m(k) \begin{pmatrix} 1 \\ 0 \\ 0 \end{pmatrix} \frac{1}{k - k_0} + O(1), \qquad k \to k_0,
	\\
m(k) \begin{pmatrix} 0 \\ 0 \\ 1 \end{pmatrix} = C_2 m(k) \begin{pmatrix} 1 \\ 0 \\ 0 \end{pmatrix} \frac{1}{k - k_0} + O(1), \qquad k \to k_0.
\end{align}
\end{subequations}
Applying the symmetry $m(k) = \mathcal{A}m(\omega k) \mathcal{A}^{-1}$, we get
\begin{align*}
& \mathcal{A}m(\omega k)\mathcal{A}^{-1}\begin{pmatrix} 0 \\ 1 \\ 0\end{pmatrix} = C_1 \mathcal{A}m(\omega k)\mathcal{A}^{-1} \begin{pmatrix} 1 \\ 0 \\ 0\end{pmatrix} \frac{1}{k - k_0}+ O(1), \qquad k \to k_0,
	\\
& \mathcal{A}m(\omega k)\mathcal{A}^{-1}\begin{pmatrix} 0 \\ 0 \\ 1 \end{pmatrix} = C_2 \mathcal{A}m(\omega k)\mathcal{A}^{-1} \begin{pmatrix} 1 \\ 0 \\ 0\end{pmatrix} \frac{1}{k - k_0}+ O(1), \qquad k \to k_0,
	\\
& \mathcal{A}^2m(\omega^2 k)\mathcal{A}^{-2}\begin{pmatrix} 0 \\ 1 \\ 0 \end{pmatrix} = C_1 \mathcal{A}^2m(\omega^2 k)\mathcal{A}^{-2} \begin{pmatrix} 1 \\ 0 \\ 0\end{pmatrix} \frac{1}{k - k_0}+ O(1), \qquad k \to k_0,
	\\
& \mathcal{A}^2m(\omega^2 k)\mathcal{A}^{-2}\begin{pmatrix} 0 \\ 0 \\ 1 \end{pmatrix} = C_2 \mathcal{A}^2m(\omega^2 k)\mathcal{A}^{-2} \begin{pmatrix} 1 \\ 0 \\ 0\end{pmatrix} \frac{1}{k - k_0}+ O(1), \qquad k \to k_0.
\end{align*}
Replacing $k$ with $\omega^2 k$ in the first and second equations, and replacing $k$ with $\omega k$ in the third and fourth equations, we can write these equations as
\begin{subequations}\label{mCDexpansions2}
\begin{align}
& m(k)\begin{pmatrix} 1 \\ 0 \\ 0\end{pmatrix} = \omega C_1 m(k) \begin{pmatrix} 0 \\ 0 \\ 1 \end{pmatrix} \frac{1}{k - \omega k_0}+ O(1), \qquad k \to \omega k_0,
	\\
& m(k)\begin{pmatrix} 0 \\ 1 \\ 0 \end{pmatrix} = \omega C_2 m(k)\begin{pmatrix} 0 \\ 0 \\ 1 \end{pmatrix} \frac{1}{k - \omega k_0}+ O(1), \qquad k \to \omega k_0,
	\\
& m(k) \begin{pmatrix} 0 \\ 0 \\ 1 \end{pmatrix} = \omega^2 C_1 m(k) \begin{pmatrix} 0 \\ 1 \\ 0\end{pmatrix} \frac{1}{k - \omega^2 k_0}+ O(1), \qquad k \to \omega^2 k_0,
	\\
& m(k)\begin{pmatrix} 1 \\ 0 \\ 0 \end{pmatrix} = \omega^2 C_2 m(k)\begin{pmatrix} 0 \\ 1 \\ 0\end{pmatrix} \frac{1}{k - \omega^2 k_0}+ O(1), \qquad k \to \omega^2 k_0.
\end{align}
\end{subequations}

Applying the symmetry $m(k) = \mathcal{B} \overline{m(\bar{k})} \mathcal{B}$ to the equations in (\ref{mCDexpansions}) and (\ref{mCDexpansions2}), we find six further residue conditions. 
Altogether we find that there are twelve relevant residue conditions given by
\begin{align*}
& [m(k)]_2 =  \frac{C_1 [m]_1}{k - k_0} + O(1) 
\;\; \text{and}\;\; 
[m(k)]_3 =  \frac{C_2 [m]_1}{k - k_0} + O(1)
&& \text{as}\; k \to k_0.
	\\
& [m(k)]_1 =  \frac{\omega C_1 [m]_3}{k - \omega k_0}+ O(1) 
\;\; \text{and}\;\;
 [m(k)]_2 =  \frac{\omega C_2 [m]_3}{k - \omega k_0}+ O(1)
 && \text{as}\; k \to \omega k_0,
	\\
& [m(k)]_3 = \frac{\omega^2 C_1 [m]_2 }{k - \omega^2 k_0}+ O(1)
\;\; \text{and}\;\;
[m(k)]_1 =  \frac{\omega^2 C_2 [m]_2}{k - \omega^2 k_0}+ O(1) 
&&\text{as}\; k \to \omega^2 k_0,
	\\
& [m(k)]_1 =  \frac{\overline{C_1} [m]_2}{k - \bar{k}_0} + O(1)
\;\; \text{and}\;\;
 [m(k)]_3 =  \frac{\overline{C_2} [m]_2}{k - \bar{k}_0} + O(1)
 &&\text{as}\; k \to \bar{k}_0.
	\\
& [m(k)]_2 = \frac{\omega^2 \overline{C_1}  [m]_3}{k - \omega^2 \bar{k}_0}+ O(1)
\;\; \text{and}\;\;
[m(k)]_1 = \frac{\omega^2 \overline{C_2}  [m]_3}{k - \omega^2 \bar{k}_0}+ O(1)
&&\text{as}\; k \to \omega^2 \bar{k}_0,
	\\
& [m(k)]_3 = \frac{\omega \overline{C_1} [m]_1}{k - \omega \bar{k}_0}+ O(1)
\;\; \text{and}\;\;
[m(k)]_2 = \frac{\omega \overline{C_2} [m]_1}{k - \omega \bar{k}_0}+ O(1)
 &&\text{as}\; k \to \omega \bar{k}_0.
\end{align*}
Hence, using that $m \to I$ as $k \to \infty$,
\begin{align*}
[m(k)]_1 =& \begin{pmatrix} 1 \\ 0 \\ 0 \end{pmatrix} 
+  \frac{\omega C_1 [m(\omega k_0)]_3}{k - \omega k_0}
+ \frac{\omega^2 C_2 [m(\omega^2 k_0)]_2}{k - \omega^2 k_0}
+ \frac{\overline{C_1} [m(\bar{k}_0)]_2}{k - \bar{k}_0}
+ \frac{\omega^2 \overline{C_2}  [m(\omega^2 \bar{k}_0)]_3}{k - \omega^2 \bar{k}_0},
	\\
 [m(k)]_2 =& \begin{pmatrix} 0\\ 1 \\ 0 \end{pmatrix} 
+\frac{C_1 [m(k_0)]_1}{k - k_0}
+ \frac{\omega C_2 [m(\omega k_0)]_3}{k - \omega k_0}
+ \frac{\omega^2 \overline{C_1}  [m(\omega^2 \bar{k}_0)]_3}{k - \omega^2 \bar{k}_0}
+\frac{\omega \overline{C_2} [m(\omega \bar{k}_0)]_1}{k - \omega \bar{k}_0},
	\\
[m(k)]_3 = & \begin{pmatrix} 0 \\ 0 \\ 1 \end{pmatrix} 
+ \frac{C_2 [m(k_0)]_1}{k - k_0}
+ \frac{\omega^2 C_1 [m(\omega^2 k_0)]_2 }{k - \omega^2 k_0}
+  \frac{\overline{C_2} [m(\bar{k}_0)]_2}{k - \bar{k}_0} 
 + \frac{\omega \overline{C_1} [m(\omega \bar{k}_0)]_1}{k - \omega \bar{k}_0}.
\end{align*}
We evaluate the first of the above equations at $k_0$ and $\omega \bar{k}_0$, the second at $\omega^2 k_0$ and $\bar{k}_0$, and the third at $\omega k_0$ and $\omega^2 \bar{k}_0$. This gives a system of equations for $[m(k_0)]_1$, $[m(\omega \bar{k}_0)]_1$, $[m(\omega^2 k_0)]_2$, $[m(\bar{k}_0)]_2$, $[m(\omega k_0)]_3$, and $[m(\omega^2 \bar{k}_0)]_3$. 
Solving this system and substituting the resulting expressions into the formula
\begin{align}\nonumber
q(x,t) =&\; \lim_{k\to \infty} k \, m_{13}(x,t,k) 
	\\\nonumber
= &\; c_{2} e^{\theta_{13}(x,t,k_0)} m_{11}(x,t,k_0)
+ \omega^2 c_1 e^{\theta_{12}(x,t,k_0)} m_{12}(x,t,\omega^2 k_0)
	\\ \label{qbreatherfinal}
& +  \overline{c_{2} e^{\theta_{13}(x,t,k_0)}} m_{12}(x,t,\bar{k}_0)
+ \omega \overline{c_1 e^{\theta_{12}(x,t,k_0)}} m_{11}(x,t,\omega \bar{k}_0),
\end{align}
where we have written $C_1 = c_1 e^{\theta_{12}(x,t,k_0)}$ and $C_{2} = c_{2} e^{\theta_{13}(x,t,k_0)}$, we obtain a family of breather soliton solutions of (\ref{3x3eq}) parametrized by $c_1, c_{2} \in \C$ and $k_0 \in D_1$. Applying the Miura-type transformation of Theorem \ref{mainth} to these solutions, we obtain solutions of the ``good'' Boussinesq equation (\ref{goodboussinesq}).
The solution in (\ref{qbreather}) is obtained from (\ref{qbreatherfinal}) in the special case when $c_1 = 0$, $c_{2} = 1$, and $k_0 = e^{\frac{\pi i}{6}}/\sqrt{3}$.

\subsection*{Acknowledgements}
Support is acknowledged from the Novo Nordisk Fonden Project, Grant 0064428, the European Research Council, Grant Agreement No. 682537, the Swedish Research Council, Grant No. 2015-05430, Grant No. 2021-04626, and Grant No. 2021-03877, the G\"oran Gustafsson Foundation, and the Ruth and Nils-Erik Stenb\"ack Foundation.

\bibliographystyle{plain}
\bibliography{is}

\end{document}